\documentclass{article}[a4paper,12pt]
\pdfoutput=1
\usepackage{bibleref}
\usepackage{caption}
\usepackage{graphicx}
\usepackage[width=150mm,top=25mm,bottom=25mm]{geometry}
\usepackage{amsmath,amsthm,latexsym,amssymb,amsfonts,latexsym}
\usepackage{setspace}
\usepackage{chemfig}
\usepackage{float}
\newtheorem{theorem}{Theorem}
\numberwithin{equation}{section}
\usepackage{xcolor,soul}
\usepackage{hyperref}
\title{Predictive ability comparison of different versions of some well known degree dependent topological indices}
\author{A. Bharali $^1$\footnote{Corresponding Author}, Mohammad Essa Nazari$^2$, Amanullah Nabavi$^3$\\$^{1,2}$Department of Mathematics, Dibrugarh University, Dibrugarh, India, 786004\\$^{3}$Department of Mathematics, Bamyan University, Bamyan, Afghanistan\\  a.bharali@dibru.ac.in$^{1*}$, nazariessa@yahoo.com$^2$, nabave786@gmail.com$^3$}
\begin{document}
\maketitle
\begin{abstract}
To avoid extensive lab work on properties of chemical compounds, QSPR/QSAR analysis for topological descriptors is a productive statistical approach to analyze various physicochemical properties of chemical compounds. Many researchers have investigated on correlation of degree-based topological descriptors. In this article we present the predictive ability of 6 well known degree dependant topological indices of 22 lower poly cyclic aromatic hydrocarbons in three versions, and we have done a comparison analysis of three versions of considered topological indices for their predictive ability.
\end{abstract}
	\textbf{Keywords: Chemical compounds, Physicochemical, Topological Indices, Predictive Ability.}
\section{Introduction}
A physicochemical characteristic of a chemical compound that is closely related to the structure, atom count, atom bonds, and atomic arrangement of the compound. Chemical graph theory uses a variety of mathematical and statistical techniques known as QSPR/QSAR models to interpret these interactions\cite{hosamani2017qspr,dehmer2017quantitative}. A molecular graph is a network whose vertices are atoms and links are bonds between atoms in chemical compounds. A molecule's numerical representation is called a molecular descriptor or topological index. A QSPR/QSAR model is a regression model that demonstrates a relationship between a molecular descriptor and a physicochemical property of a molecule; it has numerous applications in analytical chemistry, drug discovery, pharmaceutical property prediction, and material science.
Researchers investigate on models to interpret the relationship between molecular graph and the property of a molecule. Generally a topological index of a molecule calculate from underlying molecular graph a molecule. the first idea of a topological index presented by winner in 1947\cite{wiener1947structural}.\\
Examples of physicochemical properties include boiling point, enthalpy of formation, enthalpy of vaporization, critical pressure, critical volume, etc. Quantum-theoretic properties include atomic and subatomic compound characteristics. The total electronic energy, which is determined using Huckel's Molecular Orbital theory, is a prime example of quantum-theoretic properties.
The ability of topological indices to predict various physicochemical properties of chemical compounds has been the subject of numerous investigations. The predictability of several distance-based topological descriptors for the benzenoid hydrocarbons' -electronic energy was examined by Hayat et al in 2020 \cite{hayat2020distance}. Predictability of some degree-based topological indices to determine physicochemical properties of Polycyclic Aromatic Hydrocarbons (PAHs) presented by Malik et. al. \cite{malik2022correlation}. In 2022 Munjit et al \cite{chamua2022predictive} investigated on predictability of some neighborhood topological indices on PAHs. \\
Aim of this work is to investigate on predictive ability potential of different versions of some well known degree dependent topological indices in three different versions, the effect of degrees, degrees of neighbors and combination of both on molecule characteristics. In this work we utilized the graph polynomial for computing molecular descriptors. From the inception time of topological indices many works have been done on predictive ability of many degree- based and neighborhood degree sum-based topological indices individually in literature of chemical graph theory and gain vast attention of researchers but, the QSPR/QSAR analysis of closed neighborhood and comparing the predictive ability of different version of topological indices and showing the effect of degrees, degrees of neighbors and combining them are limited in the literature. This work may also be and attempt to fill up this gap. \\
The rest of this work organized as follows: The considered topological indices are defined in section 2. with some basic definitions. Computing method is presented in section 3. The predictive ability of different degree-based, neighborhood degree sum-based and closed neighborhood degree sum-based versions of topological indices for evaluating the normal boiling point of Polycyclic Aromatic Hydrocarbons presenting in section 4. A brief conclusion of this article in section 5. 

\begin{center}
\includegraphics[width=12cm,height=19cm]{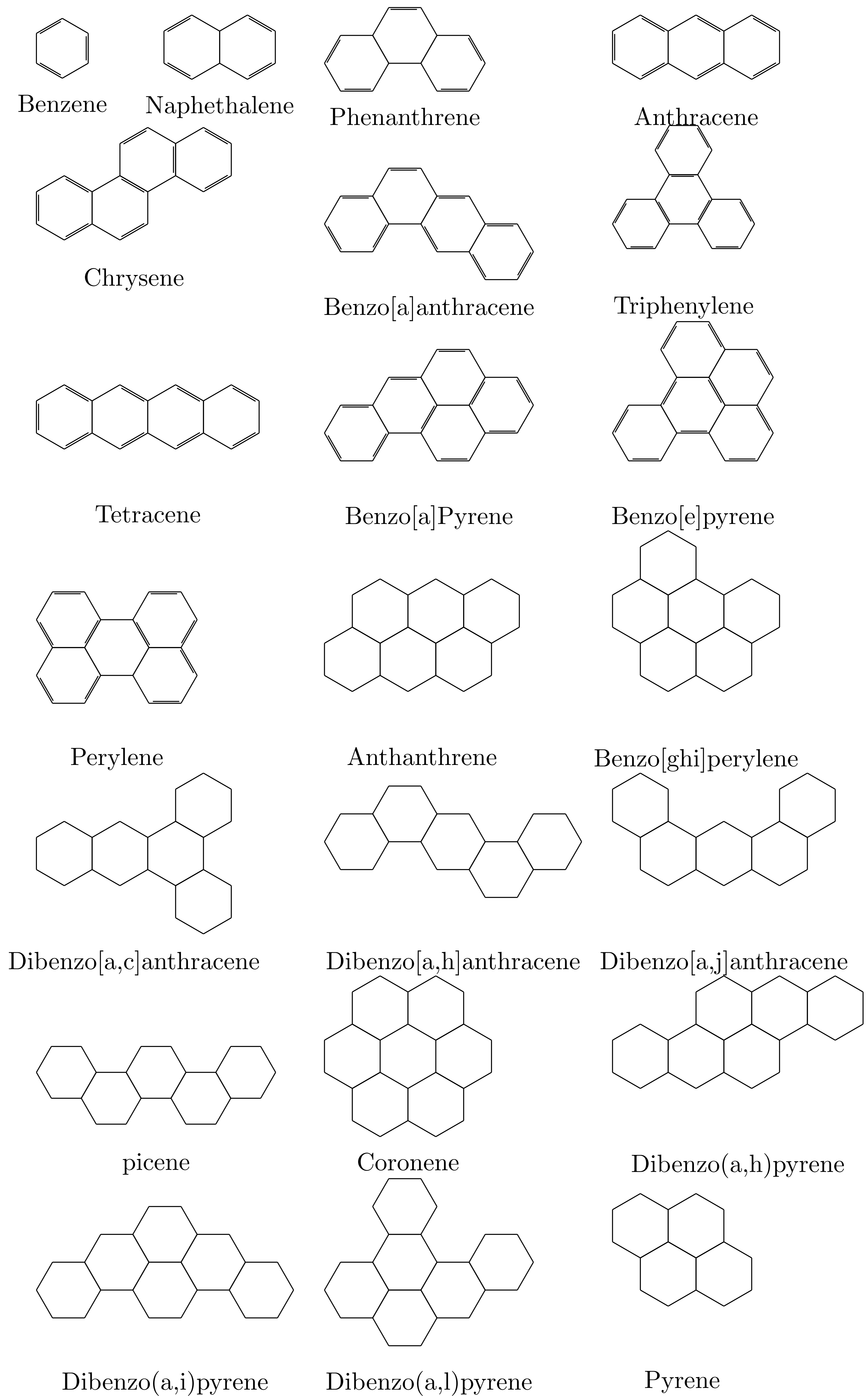}
\captionof{figure}{Molecule structures of 22 lower PAHs.}
\end{center}

\section{Preliminaries}
In this article we study on six well known topological descriptors in three versions which defined as follows:
The first and second Zagreb indices \cite{gutman2004first} are defined as,\\
\begin{equation*}
	M_{1}=\sum_{uv\in E{(\Gamma)}} (d_{u}+d_{v}),
	\end{equation*}  
\begin{equation*}
		M_{2}=\sum_{uv\in E{(\Gamma)}} (d_{u}d_{v}).
	\end{equation*} 
The Randi\'{c} index  \cite{randic1975characterization} is defined as,
\begin{equation*}
     R=\sum_{uv\in E{(\Gamma)}} \frac{1}{\sqrt{d_{u}d_{v}}}.
\end{equation*}
The Atom Bond Connectivity index \cite{estrada1998atom} is defined as,
\begin{equation*}
		ABC=\sum_{uv\in E{(\Gamma)}} \sqrt{\frac{d_{u}+d_{v}-2}{d_{u}d_{v}}}.
\end{equation*}
The Geometric Arithmetic index \cite{vukivcevic2009topological} is defined as,
\begin{equation*}
	 GA =\sum_{uv\in E{(\Gamma)}} \frac{2\sqrt{d_{u}d_{v}}}{d_{u}+d_{v}}.
\end{equation*}
The Redefined second version of Zagreb index \cite{pradeep2017redefined} is defined as,
\begin{equation*}
	ReZG_{2} =\sum_{uv\in E{(\Gamma)}} \frac{d_{u}d_{v}}{d_{u}+d_{v}}.
\end{equation*}
The third version of Zagreb index \cite{ghorbani2013third} is defined as,
\begin{equation*}
	M_{1}^{*}=\sum_{uv\in E{(\Gamma)}} (S_{(u)}+S_{(v)}).	
\end{equation*}
The neighborhood second Zagreb index \cite{mondal2019some} is defined as,
\begin{equation*}
	M_{2}^{*}=\sum_{uv\in E{(\Gamma)}} (S_{(u)}d_{(v)}).
\end{equation*}
The neighborhood Randi\'{c} index \cite{abubakar2022neighborhood} is defined as, 
\begin{equation*}
	R^{*}=\sum_{uv\in E{(\Gamma)}} \frac{1}{\sqrt{S_{(u)}S_{(v)}}}.
\end{equation*}
The neighborhood version of ABC index \cite{abubakar2022neighborhood} is defined as,
\begin{equation*}
	NABC=\sum_{uv\in E{(\Gamma)}} \sqrt{\frac{S_{(u)}+S_{(v)}-2}{S_{(u)}S_{(v)}}}.
\end{equation*}
The neighborhood version of GA index \cite{abubakar2022neighborhood} is defined as,
\begin{equation*}
	 NGA =\sum_{uv\in E{(\Gamma)}} \frac{2\sqrt{S_{(u)}S_{(v)}}}{S_{(u)}+S_{(v)}}.
\end{equation*}
The neighborhood redefined second Zagreb index is defined as,
\begin{equation*}
		NReZG_{2} =\sum_{uv\in E{(\Gamma)}} \frac{S_{(u)}S_{(v)}}{S_{(u)}+S_{(v)}}.
	\end{equation*}
Where the neighborhood degree sum $S_{u}$ of a vertex $u$ the sum of degrees all neighborhood vertices of degree $u$.\\
The closed neighborhood form of these indices as follows:
\begin{equation*}
		N^{c}M_{1}=\sum_{uv\in E{(\Gamma)}} (S_{[u]}+S_{[v]}).
\end{equation*}
\begin{equation*}
		N^{c}M_{2}=\sum_{uv\in E{(\Gamma)}} (S_{[u]}d_{[v]}).
\end{equation*}
\begin{equation*}
			N^{c}R=\sum_{uv\in E{(\Gamma)}} \frac{1}{\sqrt{S_{[u]}S_{[v]}}}.
\end{equation*}
\begin{equation*}
			N^{c}ABC=\sum_{uv\in E{(\Gamma)}} \sqrt{\frac{S_{[u]}+S_{[v]}-2}{S_{[u]}S_{[v]}}}.
\end{equation*}
\begin{equation*}
	 	N^{c}GA =\sum_{uv\in E{(\Gamma)}} \frac{2\sqrt{S_{[u]}S_{[v]}}}{S_{[u]}+S_{[v]}}.
\end{equation*}
\begin{equation*}
			N^{c}ReZG_{2} =\sum_{uv\in E{(\Gamma)}} \frac{S_{[u]}S_{[v]}}{S_{[u]}+S_{[v]}}.
\end{equation*}
Where $S_{[u]}$ is the sum of degrees of closed neighborhood vertices of degree $u$.
\section{Computational method}
In this section we provide detail for computational method by utilizing graph polynomials as a tool for computing topological indices. For degree dependent versions of topological indices we use  different versions of graph polynomials, for example; to compute degree-based, neighborhood degree sum-based we use M-polynomial \cite{deutsch2014m} and NM-polynomial \cite{mondal2021neighborhood}. By inspiring the mentioned polynomials, for closed neighborhood degree sum-based version of topological indices \cite{ravi2022closed} we can use CNM-polynomial. The formula of underlying graph polynomials for a graph $\Gamma$ shown as follows respectively:

\begin{align*}
P(\Gamma;x,y)=M(\Gamma;x,y)=\sum_{i\leq j}\chi _{ij}x^{i}y^{j}
\end{align*}
where, $\chi_{ij},\; i,j\geq1 $ is the number of edges $e=uv $ in $ \Gamma$ such that $\{ d_{u}, d_{v} \}=\{i,j\}.$

\begin{align*}
P(\Gamma;x,y)=NM(\Gamma;x,y)=\sum_{i\leq j}\chi _{(i,j)}x^{i}y^{j}
\end{align*}
 
 where, $\chi _{(i,j)}, \; i,j\geq 1 $ is the number of edges $e=uv$ in $ \Gamma $ such that $\{S_{(u)},S_{(v)}\}=\{i,j\}$ and $S_{(u)}=\sum_{v\in N_{\Gamma}(u)}d_{v}.$ And 
\begin{align*}
P(\Gamma;x,y)=CNM(\Gamma;x,y)=\sum_{i\leq j}\chi_{[i,j]}x^{i}y^{j}
\end{align*}
where, $\chi _{[i,j]}, \; i,j\geq 1 $ is the number of edges $e=uv$ in $ \Gamma $ such that $\{S_{[u]},S_{[v]}\}=\{i,j\}$ and $S_{[u]}=\sum_{v\in N_{\Gamma}[u]}d_{v}.$

\begin{table}[ht]
	\centering
	\resizebox{\textwidth}{!}
	{
		\begin{tabular}{lllll}
			\hline
			\textbf{DB} & \textbf{NDSB} &\textbf{CNDSB}& \textbf{$f(x,y)$} & \textbf{Derivation from P(G;x,y)} \\
		\hline
			$M_{1}$ & $M_{1}^{*}$ & $N^{c}M_{1}$& $x+y$ & $(D_{x}+D_{y})(P(G;x,y))|_{x=y=1}$ \\
		
			$M_{2}$ & $M_{2}^{*}$ &$N^{c}M_{2}$& $xy$ & $(D_{x}D_{y})(P(G;x,y))|_{x=y=1}$ \\
		
			$R$ & $R^{*}$ &$N^{c}R$&  $\frac{1}{\sqrt{xy}}$ & $S_{x}^{\frac{1}{2}}S_{y}^{\frac{1}{2}}(P(G;x,y))|_{x=y=1}$ \\
		
			$ABC$ & $NABC(G)$ &$N^{c}ABC$& $\sqrt{\frac{x+y-2}{xy}}$ & $D_{x}^{\frac{1}{2}}Q_{-2}JS_{y}^{\frac{1}{2}}S_{x}^{\frac{1}{2}}(P(G;x,y))|_{x=y=1}$ \\
		
			$GA$ & $NGA$ &$N^{c}GA$& $\frac{2\sqrt{xy}}{x+y}$ & $2S_{x}JD_{y}^{\frac{1}{2}}D_{x}^{\frac{1}{2}}P(G;x,y)|_{x=y=1}$ \\
		
			$ReZG_{2}$ & $NReZG_{2}$ &$N^{c}ReZG_{2}$&$\frac{xy}{x+y}$& $S_{x}JD_{x}D_{y}P(G;x,y)|_{x=y=1} $\\
		
			\hline
		\end{tabular}
	}
	\caption{Derivation formula for Degree-Based(DB), Neighborhood Degree Sum-Based(NDSB) and Closed Neighborhood Degree Sum-Based(CNDSB) versions of 6 well known degree dependent Topological Indices(TIs).}
\end{table}
Where,
\begin{align*}
	 D_{x}&=x\frac{\partial(\gamma(x,y))}{\partial x},\;\; D_{y}&=y\frac{\partial(\gamma(x,y))}{\partial y},\;\; S_{x}&=\int_{0}^{x}\frac{\gamma(x,y)_{x=t}}{t}dt,\;\; S_{y}&=\int_{0}^{y}\frac{\gamma(x,y)_{y=t}}{t}dt,\;\;
	  J(\gamma(x,y))&=\gamma(x,y)_{x=y},
	  \end{align*}
  \begin{align*}
	  Q_{\alpha}\gamma(x,y)&=x^{\alpha}\gamma(x,y),
	  D_{x}^{\frac{1}{2}}\gamma(x,y)&=\sqrt{x\frac{\partial\gamma(x,y)}{\partial x}}. \sqrt{\gamma(x,y)} S_{x}^{\frac{1}{2}}\gamma(x,y)&=\sqrt{\int_{0}^{x}\frac{\gamma(x,y)_{x=t}}{t}dt}. \sqrt{\gamma(x,y)}.
\end{align*}
	By doing edge partition with respect to degree vertex, neighborhood degree and closed neighborhood degree of 22 lower PAHs and using the derivation formula from table 1. for corresponding polynomial we compute the desire results.

\begin{theorem}
	Consider the molecular graph $\Gamma$ for Naphthalene; then, $M_{1}=50$, $M_{2}=57$, $ R=4.96633$, $ABC=7.73773$, $GA=6$, $ReZG_{2}=6$, $M_{1}^{*}=114$, $M_{2}^{*}=301$, $ R^{*}=2.21341$, $NABC=6.22414$, $NGA=10.9193$, $NReZG_{2}=28.0556$, $N^{c}M_{1}=164$, $N^{c}M_{2}=620$, $N^{c}R=1.52864$, $N^{c}ABC=5.37706$, $N^{c}GA=10.9881$ and $N^{c}ReZG_{2}=40.3937$.
\end{theorem}
\begin{proof}
From Figure 1, we consider the molecular graph of naphthalene, it is obvious that there are 10 vertices and 11 edges. based on degrees of end vertices there are three types, based on neighborhood degree three types of edges counting 11. They are as follows.
\begin{align*}
	E_{2,2}&=\{e=uv\in E(\Gamma)|d_{u}=2, d_{v}=2\},\\
	E_{2,3}&=\{e=uv\in E(\Gamma)|d_{u}=2, d_{v}=3\},\\
	E_{3,3}&=\{e=uv\in E(\Gamma)|d_{u}=2, d_{v}=3\}.
	\end{align*}
Edge partition for neighborhood degree of end vertices are as.
\begin{align*}
	E_{4,4}&=\{e=uv\in E(\Gamma)|S_{(u)}=4, S_{(v)}=4\},\\
	E_{4,5}&=\{e=uv\in E(\Gamma)|S_{(u)}=4, S_{(v)}=5\},\\
	E_{5,7}&=\{e=uv\in E(\Gamma)|S_{(u)}=5, S_{(v)}=7\},\\
	E_{7,7}&=\{e=uv\in E(\Gamma)|S_{(u)}=7, S_{(v)}=7\}.
	\end{align*}
Edge partition for closed neighborhood degree of end vertices are sa.
\begin{align*}
E_{6,6}&=\{e=uv\in E(\Gamma)|S_{[u]}=6, S_{[v]}=6\},\\
E_{6,7}&=\{e=uv\in E(\Gamma)|S_{[u]}=6, S_{[v]}=7\},\\
E_{7,10}&=\{e=uv\in E(\Gamma)|S_{[u]}=7, S_{[v]}=10\},\\
E_{10,10}&=\{e=uv\in E(\Gamma)|S_{[u]}=10, S_{[v]}=10\}.
\end{align*}
such that 
\begin{align*}
	|E_{2,2}|&=6,\\
	|E_{2,3}|&=4,\\
	|E_{3,3}|&=1,\\
	|E_{(4,4)}|&=2,\\
	|E_{(4,5)}|&=4,\\
	|E_{(5,7)}|&=4,\\
	|E_{(7,7)}|&=1,\\
	|E_{[6,6]}|&=2,\\
	|E_{[6,7]}|&=4,\\
	|E_{[7,10]}|&=4,\\
	|E_{[10,10]}|&=1.
	\end{align*}
Considering the edge partition and using the corresponding polynomials as,
\begin{align*}
	M(\Gamma;x,y)&=\sum_{i\leq j}\chi_{ij}x^{i}y^{j}=6x^{2}y^{2}+4x^{2}y^{3}+x^{3}y^{3},\\
	NM(\Gamma;x,y)&=\sum_{i\leq j}\chi_{(ij)}x^{i}y^{j}=2x^{4}y^{4}+4x^{4}y^{5}+4x^{5}y^{7}+x^{7}y^{7},\\
		CNM(\Gamma;x,y)&=\sum_{i\leq j}\chi_{[ij]}x^{i}y^{j}=2x^{6}y^{6}+4x^{6}y^{7}+4x^{7}y^{10}+x^{10}y^{10}.
\end{align*}
now by applying the derivation formula from Table 1 we can obtain the desire results.

\end{proof}
The rest of the result can be obtained in this manner.

\begin{table}[ht]
	\centering
	\caption{Numerical values of Degree-Based versions of topological descriptors.}
	\resizebox{\textwidth}{!}
	{
		\begin{tabular}{lccccccc}
			\hline
			
			Molecules & $ M_1 $   & $ M_2 $   & $ R $     & $ ABC $   & $ GA $    & $ ReZG_2 $ & $ BP $ \\
			\hline
			
			Benzene & 24    & 24    & 3     & 4.24264 & 6     & 6     & 80.1 \\
			
			Naphthalene & 50    & 57    & 4.96633 & 7.73773 & 10.9192 & 12.3  & 218 \\
			
			Phenanthrene & 76    & 91    & 6.94949 & 11.1924 & 15.8788 & 18.7  & 338 \\
			
			Anthracene & 76    & 90    & 6.93265 & 11.2328 & 15.8384 & 18.6  & 340 \\
			
			Chrysene & 102   & 125   & 8.93265 & 14.647 & 20.8384 & 25.1  & 431 \\
			
			Benzo[a]anthracene & 102   & 124   & 8.91582 & 14.6875 & 20.798 & 25    & 425 \\
			
			Triphenylene & 102   & 126   & 8.94949 & 14.6066 & 20.8788 & 25.2  & 429 \\
			
			Tetracene & 102   & 123   & 8.89898 & 14.7279 & 20.7576 & 24.9  & 440 \\
			
			Benzo[a]pyrene & 120   & 151   & 9.91582 & 16.6875 & 23.798 & 29.5  & 496 \\
			
			Benzo[e]pyrene & 120   & 152   & 9.93265 & 16.647 & 23.8384 & 29.6  & 493 \\
			
			Perylene & 120   & 152   & 9.93265 & 16.647 & 23.8384 & 29.6  & 497 \\
			
			Anthanthrene & 138   & 177   & 10.899 & 18.7279 & 26.7576 & 33.9  & 547 \\
			
			Benzo[ghi]preylene & 138   & 178   & 10.9158 & 18.6875 & 26.798 & 34    & 542 \\
			
			Dibenzo[a,c]anthracene & 128   & 159   & 10.9158 & 18.1017 & 25.798 & 31.5  & 535 \\
			
			Dibenzo[a,h]anthracene & 128   & 158   & 10.899 & 18.1421 & 25.7576 & 31.4  & 535 \\
			
			Dibenzo[a,j]anthracene & 128   & 158   & 10.899 & 18.1421 & 25.7576 & 31.4  & 531 \\
			
			Picene & 128   & 159   & 10.9158 & 18.1017 & 25.798 & 31.5  & 519 \\
			
			Coronene & 156   & 204   & 11.899 & 20.7279 & 29.7576 & 38.4  & 590 \\
			
			Dibenzo(a,b)pyrene & 146   & 185   & 11.899 & 20.1421 & 28.7576 & 35.9  & 596 \\
			
			Dibenzo(a,i)pyrene & 146   & 185   & 11.899 & 20.1421 & 28.7576 & 35.9  & 594 \\
			
			Dibenzo(a,l)pyrene & 146   & 186   & 11.9158 & 20.1017 & 28.798 & 36    & 595 \\
			
			Pyrene & 94    & 117   & 7.93265 & 13.2328 & 18.8384 & 23.1  & 393 \\
			\hline
			
	\end{tabular}}%
	\label{tab:addlabel1}%
\end{table}%

% Table generated by Excel2LaTeX from sheet 'Neighborhood sum based'
\begin{table}[ht]
	\centering
	\caption{Numerical values of Neighborhood degree sum-based topological descriptors.}
	\resizebox{\textwidth}{!}{
		\begin{tabular}{lccccccc}
			
			\hline	
			Molecules & $ M_1^* $ & $ M_2^* $ & $ R^* $   & $ NABC $  & $ NGA $   & $ NReZG_2 $ & $ BP $ \\
			\hline
			Benzene & 48    & 96    & 1.5   & 3.67423 & 6     & 12    & 80.1 \\
			
			Naphthalene & 114   & 301   & 2.21341 & 6.22414 & 10.9193 & 28.0556 & 218 \\
			
			Phenanthrene & 182   & 533   & 2.97904 & 8.77509 & 15.8609 & 44.6761 & 338 \\
			
			Anthracene & 180   & 518   & 2.97348 & 8.76608 & 15.9074 & 44.4786 & 340 \\
			
			Chrysene & 250   & 766   & 3.74527 & 11.325 & 20.8069 & 61.3299 & 431 \\
			
			Benzo[a]anthracene & 248   & 751   & 3.74006 & 11.3154 & 20.8547 & 61.1367 & 425 \\
			
			Triphenylene & 252   & 792   & 3.79032 & 11.3394 & 20.8009 & 61.7949 & 429 \\
			
			Tetracene & 246   & 735   & 3.73355 & 11.308 & 20.8956 & 60.9017 & 440 \\
			
			Benzo[a]pyrene & 302   & 982   & 4.049 & 12.6543 & 23.8039 & 74.2854 & 496 \\
			
			Benzo[e]pyrene & 304   & 1004  & 4.07401 & 12.6667 & 23.7798 & 74.6533 & 493 \\
			
			Perylene & 304   & 1002  & 4.07081 & 12.6725 & 23.7639 & 74.5683 & 497 \\
			
			Anthanthrene & 354   & 1199  & 4.35315 & 13.9828 & 26.8044 & 87.2703 & 547 \\
			
			Benzo[ghi]preylene & 356   & 1218  & 4.36089 & 13.988 & 26.7745 & 87.5967 & 542 \\
			
			Dibenzo[a,c]anthracene & 318   & 1011  & 4.5523 & 13.8781 & 25.8005 & 78.293 & 535 \\
			
			Dibenzo[a,h]anthracene & 316   & 984   & 4.50664 & 13.8647 & 25.802 & 77.7947 & 535 \\
			
			Dibenzo[a,j]anthracene & 316   & 984   & 4.50664 & 13.8647 & 25.802 & 77.7947 & 531 \\
			
			Picene & 318   & 999   & 4.5115 & 13.8749 & 25.7529 & 77.9838 & 519 \\
			
			Coronene & 408   & 1434  & 4.65097 & 15.3035 & 29.7851 & 100.625 & 590 \\
			
			Dibenzo(a,b)pyrene & 370   & 1221  & 4.83627 & 15.2045 & 28.7759 & 91.0986 & 596 \\
			
			Dibenzo(a,i)pyrene & 370   & 1221  & 4.83627 & 15.2045 & 28.7759 & 91.0986 & 594 \\
			
			Dibenzo(a,l)pyrene & 373   & 1256  & 4.86257 & 15.2071 & 28.7643 & 91.7876 & 595 \\
			
			Pyrene & 234   & 743   & 3.26174 & 10.104 & 18.832 & 57.4722 & 393 \\
			\hline
			
	\end{tabular}}%
	\label{tab:addlabel2}%
\end{table}%

% Table generated by Excel2LaTeX from sheet 'Closed Neighborhood'
\begin{table}[ht]
	\centering
	\caption{Numerical values of Closed Neighborhood degree sum-based version of topological descriptors. }
	\resizebox{\textwidth}{!}{
		\begin{tabular}{lccccccc}
			\hline
			Molecules & \multicolumn{1}{l}{$ N^cM_1 $} & \multicolumn{1}{l}{$ N^cM_2 $} & \multicolumn{1}{l}{$ N^cR $} & \multicolumn{1}{l}{$ N^cABC $} & \multicolumn{1}{l}{$ N^cGA $} & \multicolumn{1}{l}{$ N^cReZG_2 $} & \multicolumn{1}{l}{$ BP $} \\
			\hline
			Benzene & 72    & 216   & 1     & 3.16228 & 6     & 18    & 80.1 \\
			Naphthalene & 164   & 620   & 1.52864 & 5.37706 & 10.9881 & 40.3937 & 218 \\
			Phenanthrene & 258   & 1064  & 2.08102 & 7.59712 & 15.8731 & 63.4254 & 338 \\
			Anthracene & 256   & 1040  & 2.07585 & 7.59018 & 15.9006 & 63.1714 & 340 \\
			Chrysene & 352   & 1509  & 2.63361 & 9.81679 & 20.8231 & 86.481 & 431 \\
			Benzo[a]anthracene & 350   & 1485  & 2.62859 & 9.80962 & 20.8513 & 86.2302 & 425 \\
			Triphenylene & 354   & 1548  & 2.65504 & 9.82614 & 20.8322 & 87.0513 & 429 \\
			Tetracene & 348   & 1460  & 2.62307 & 9.8033 & 20.8758 & 85.9492 & 440 \\
			Benzo[a]pyrene & 434   & 1939  & 3.03541 & 11.5111 & 24.8171 & 106.914 & 496 \\
			Benzo[e]pyrene & 424   & 1936  & 2.88199 & 10.9923 & 23.8093 & 104.354 & 493 \\
			Perylene & 424   & 1934  & 2.88091 & 10.9942 & 23.8014 & 104.293 & 497 \\
			Anthanthrene & 492   & 2298  & 3.10405 & 12.151 & 26.813 & 121.369 & 547 \\
			Benzo[ghi]preylene & 494   & 2326  & 3.11002 & 12.1565 & 26.7941 & 121.716 & 542 \\
			Dibenzo[a,c]anthracene & 446   & 1970  & 3.20297 & 12.038 & 25.8133 & 109.883 & 535 \\
			Dibenzo[a,h]anthracene & 444   & 1930  & 3.18133 & 12.0291 & 25.802 & 109.289 & 535 \\
			Dibenzo[a,j]anthracene & 444   & 1930  & 3.18133 & 12.0291 & 25.802 & 109.289 & 531 \\
			Picene & 446   & 1954  & 3.18621 & 12.0365 & 25.7731 & 109.537 & 519 \\
			Coronene & 564   & 2718  & 3.33914 & 13.3188 & 29.7868 & 139.139 & 590 \\
			Dibenzo(a,b)pyrene & 516   & 2356  & 3.4299 & 13.2042 & 28.7881 & 127.146 & 596 \\
			Dibenzo(a,i)pyrene & 516   & 2356  & 3.4299 & 13.2042 & 28.7881 & 127.146 & 594 \\
			Dibenzo(a,l)pyrene & 519   & 2407  & 3.44301 & 13.2052 & 28.7918 & 127.936 & 595 \\
			Pyrene & 328   & 1450  & 2.30759 & 8.76395 & 18.846 & 80.6824 & 393 \\
			\hline
		\end{tabular}%
	}
	\label{tab:addlabel3}%
\end{table}%
\section{Data analysis and discussion}
In this section we provide an analysis of our computed results. The correlation coefficient values of DB, NDSB and CNDSB form of considered topological indices with normal boiling points of 22 PAHs are obtained in tables 5, 6 and 7 respectively. In all three forms of underlying TIs are high correlated with BP but  $R$, $ABC$ and $GA$ as compared to the other topological indices have excellent values of Correlation coefficient with normal boiling point. In table 5 the  $R$, $ABC$ and $GA$ have CC values $0.996293$, $0.996243$ and $0.99554$ respectively. In Tables 6 CC values of $NABC$, $NGA$ and $R^{*}$ are $0.99664, 0.99564$ and $0.991817$ respectively. In table 7 CC values of  $N^{c}ABC$, $N^{c}GA$ and $N^{c}R$ are $0.996343, 0.99544$ and $0.993579$ respectively.
We use linear regression model for considered topological indices. we make scatter plot of all three versions of given topological indices in Figures 2, 3 and 4. 
\subsection{Comparison of predictive ability of DB, NDSB and CNDSB forms of TIs} 
In this section we compare the predictive potential of given topological indices for three versions. In this comparison analysis we can indicate that the DB form of given topological indices have better predictive ability potential as compered to NDSB and CNDSB forms. The second best form for predictive is CNDSB. From this analysis we may find that the degree of a vertex is more influence than the degree of neighborhood vertices or we can say the first neighborhood is more important than the second and third neighborhood in order of distance. We have a graphical comparison of these three forms of given topological indices in Figure 5.

\begin{figure}[ht]
	\centering
	\includegraphics[width=5.7cm,height=3.4cm]{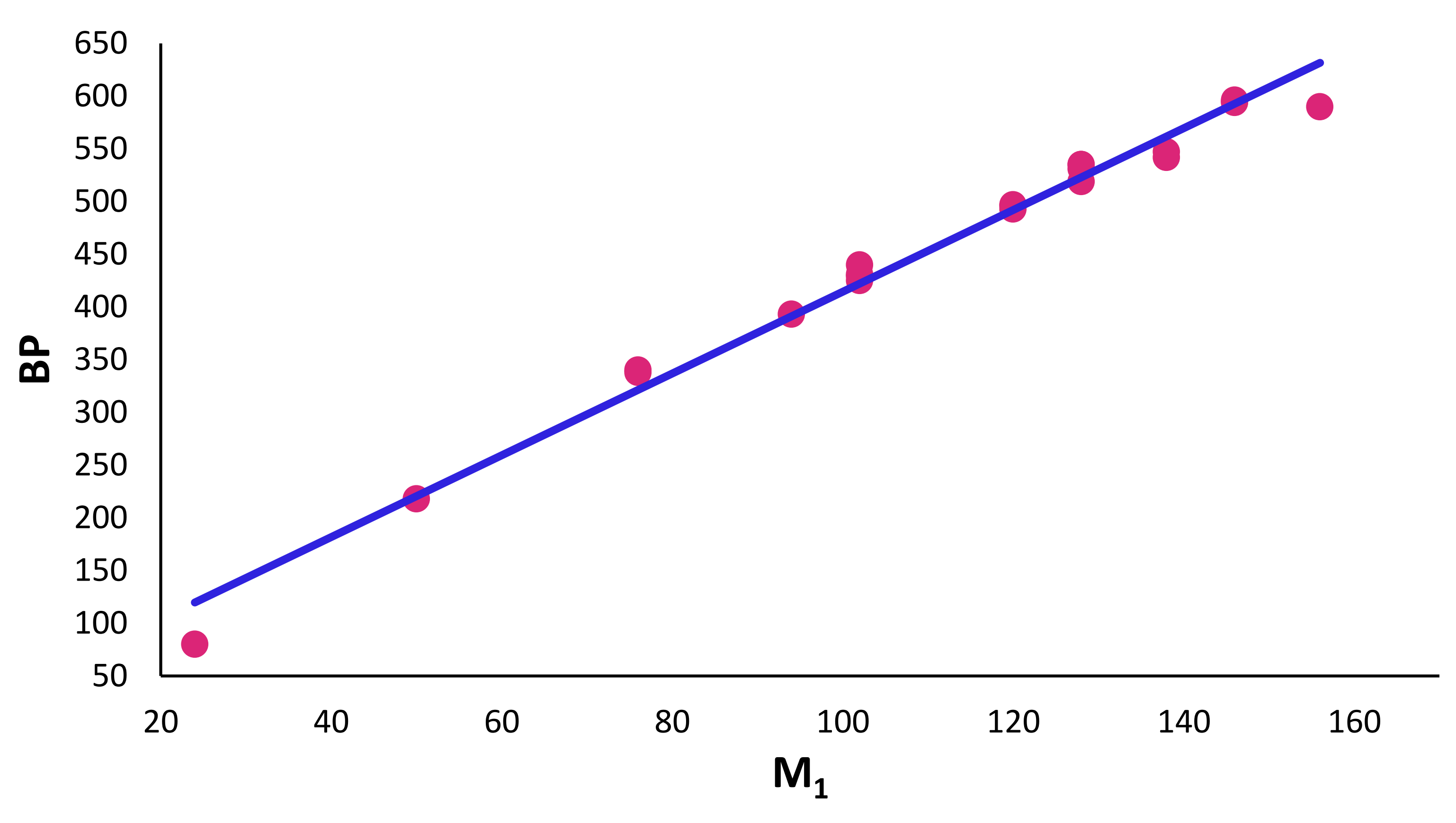}\hspace{1.5cm}
	\includegraphics[width=5.7cm,height=3.4cm]{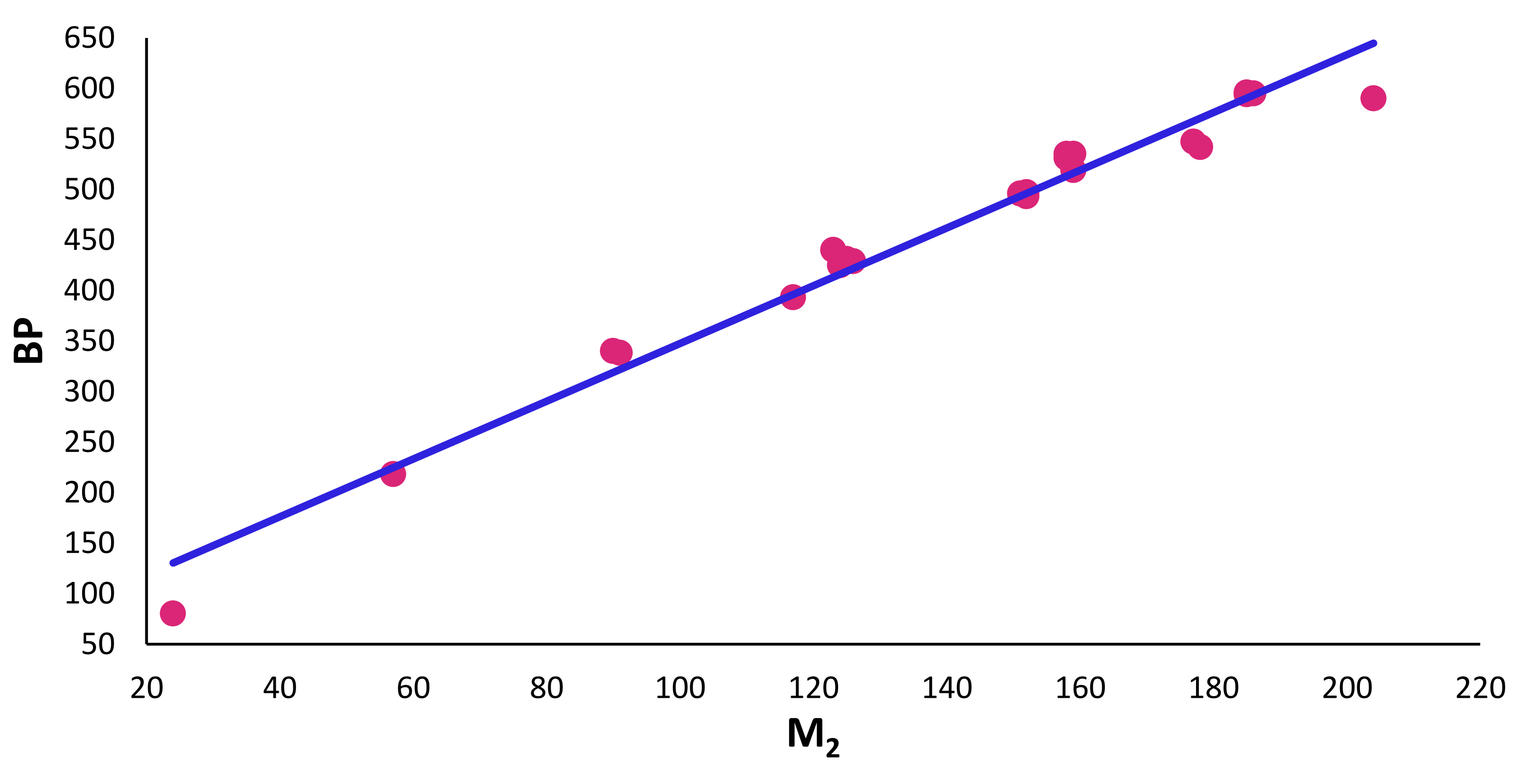}\vspace{0.5cm}
	\includegraphics[width=5.7cm,height=3.4cm]{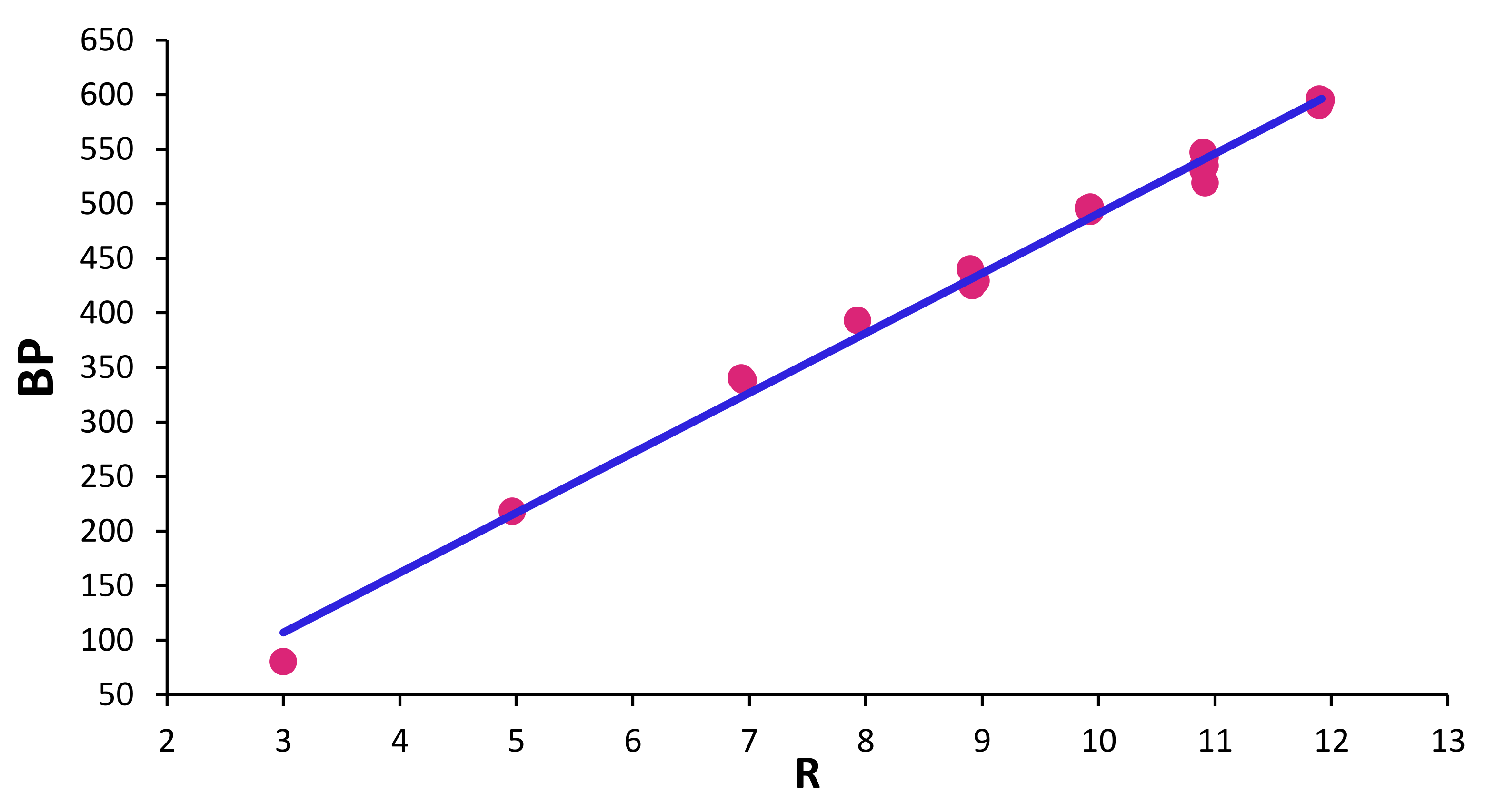}\hspace{1.5cm}
	\includegraphics[width=5.7cm,height=3.4cm]{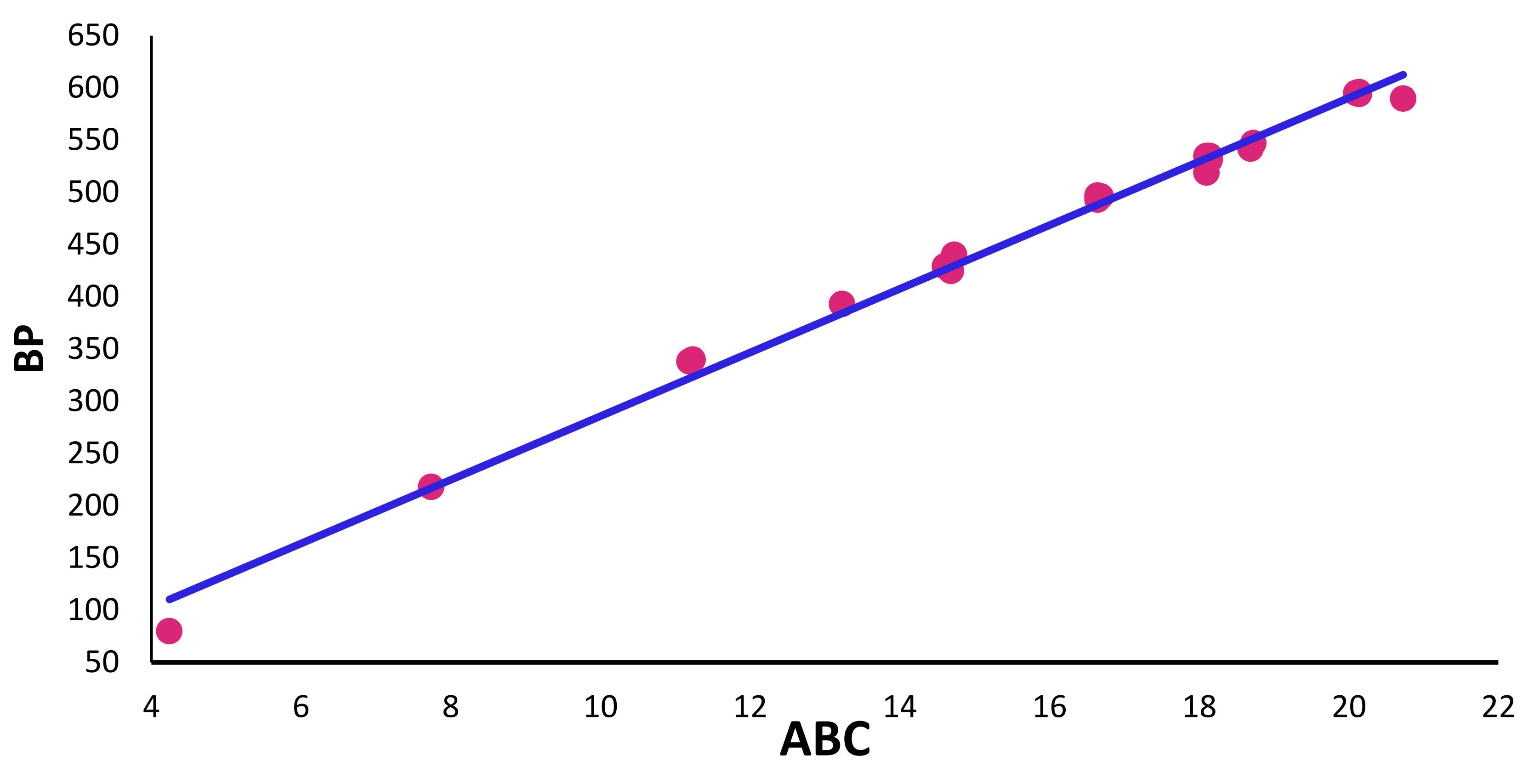}\vspace{0.5cm}
	\includegraphics[width=5.7cm,height=3.4cm]{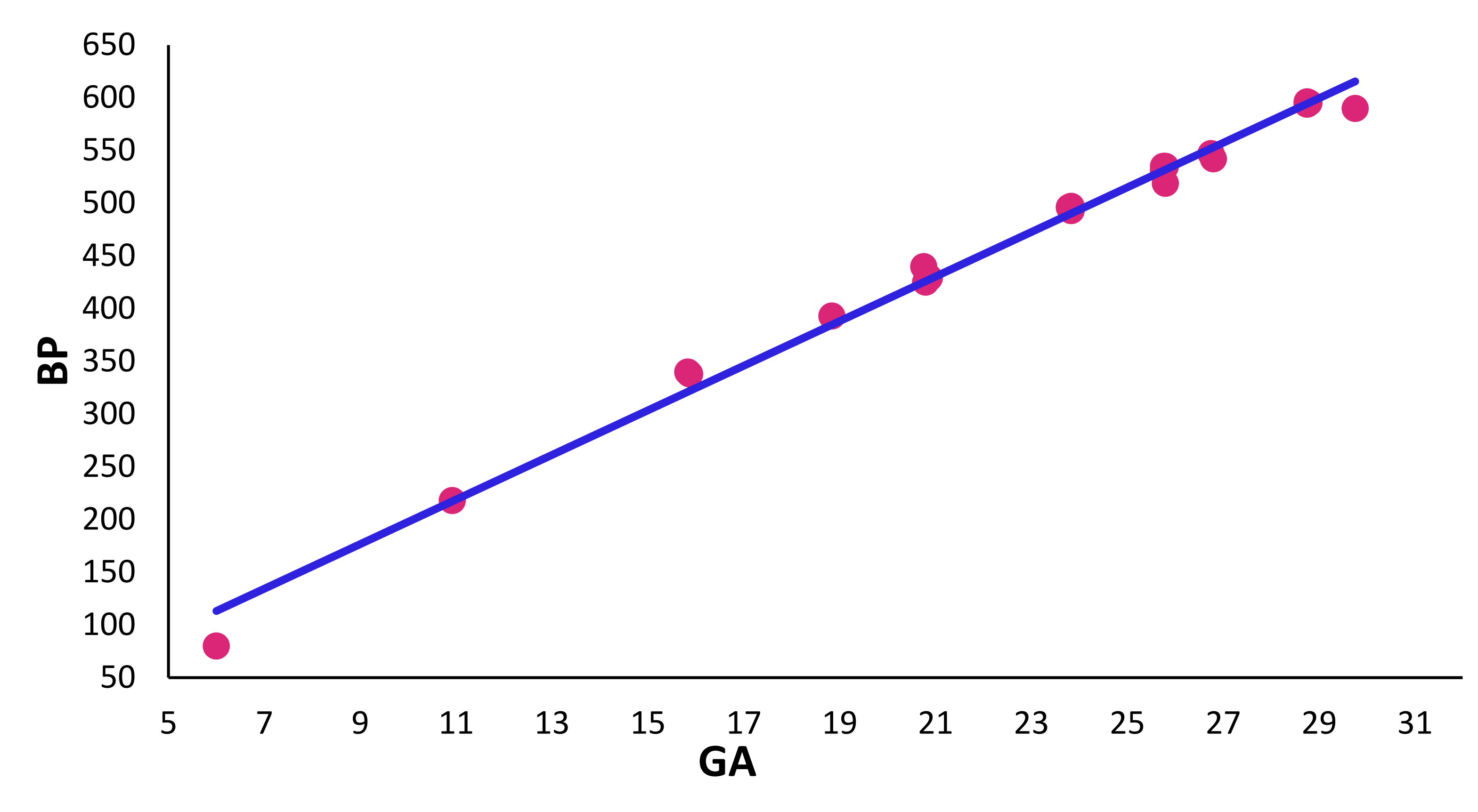}\hspace{1.5cm}
	\includegraphics[width=5.7cm,height=3.4cm]{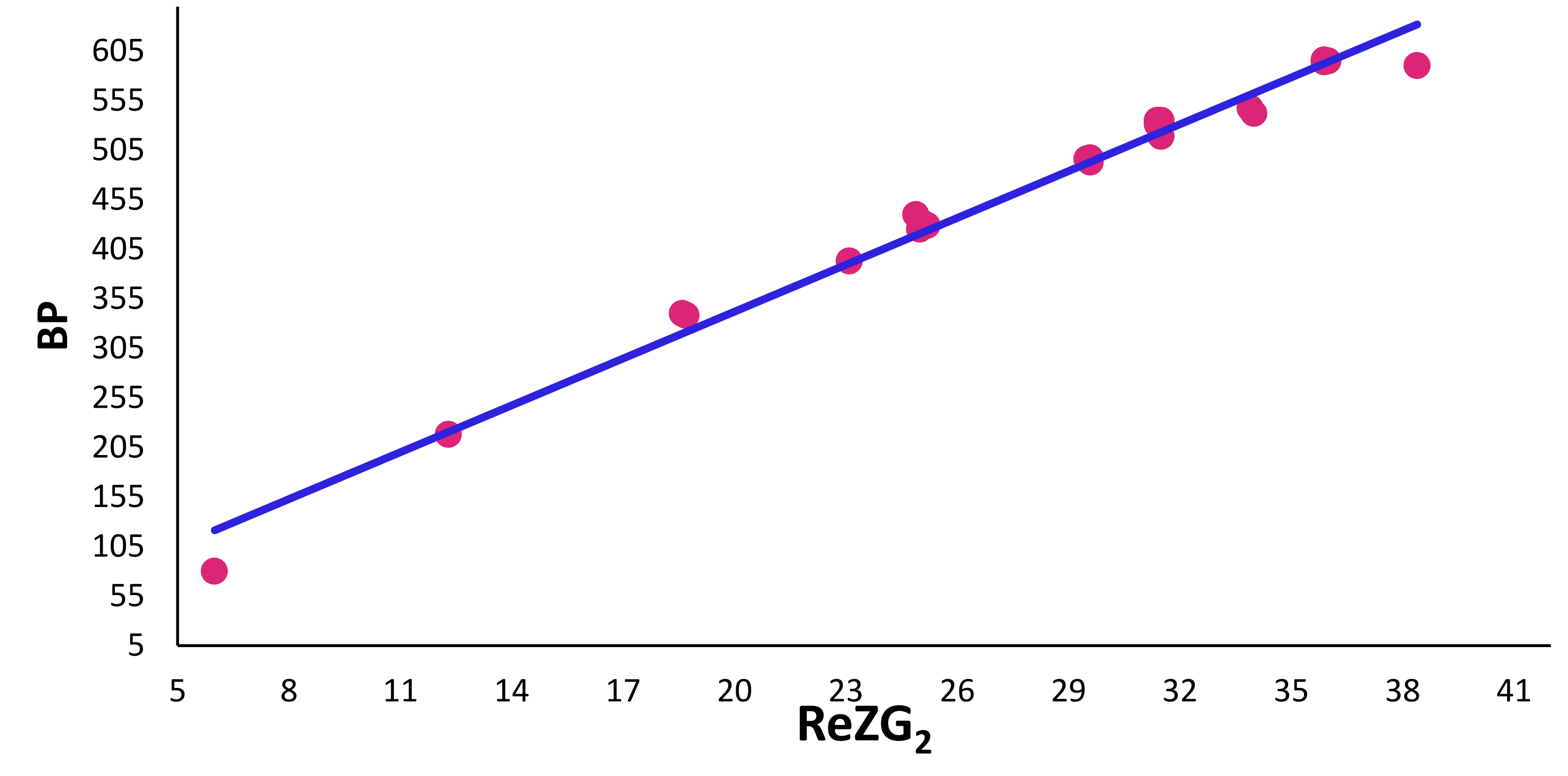}
	\caption{Graphical representation regression model between Boiling point and DB version of TIs.}
\end{figure}
\begin{figure}[ht]
	\centering
	\includegraphics[width=5.7cm,height=3.4cm]{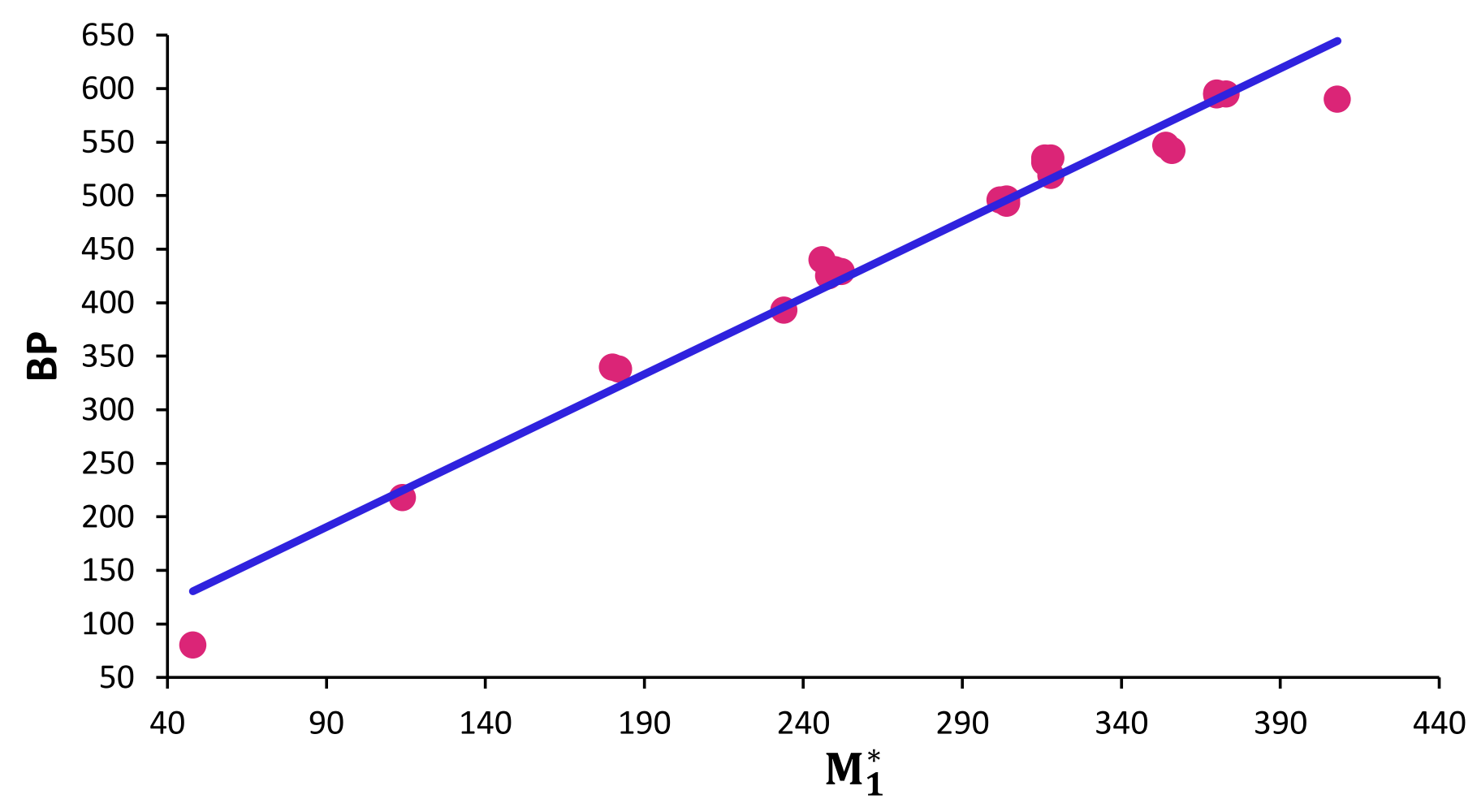}\hspace{1.5cm}
	\includegraphics[width=5.7cm,height=3.4cm]{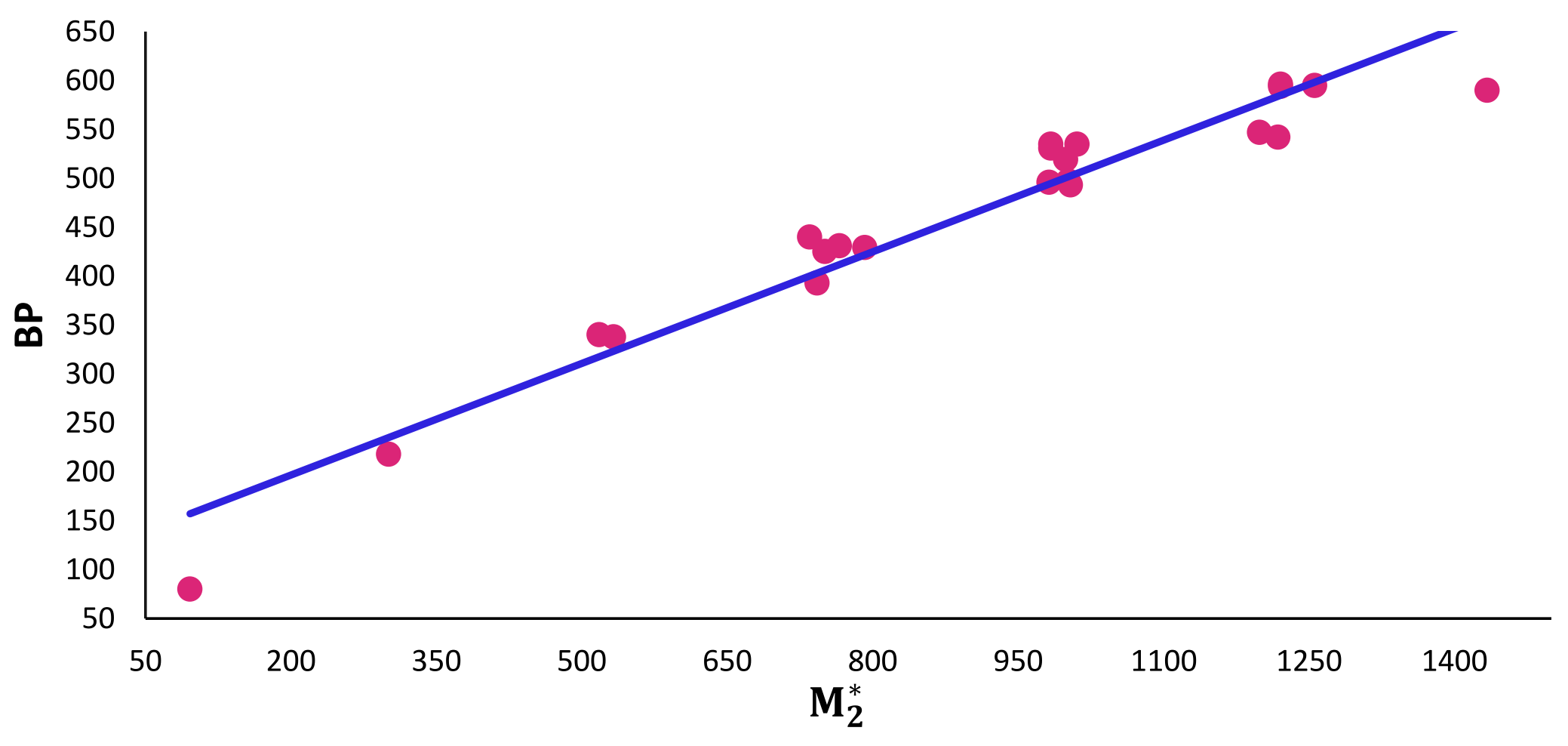}\vspace{0.5cm}
	\includegraphics[width=5.7cm,height=3.4cm]{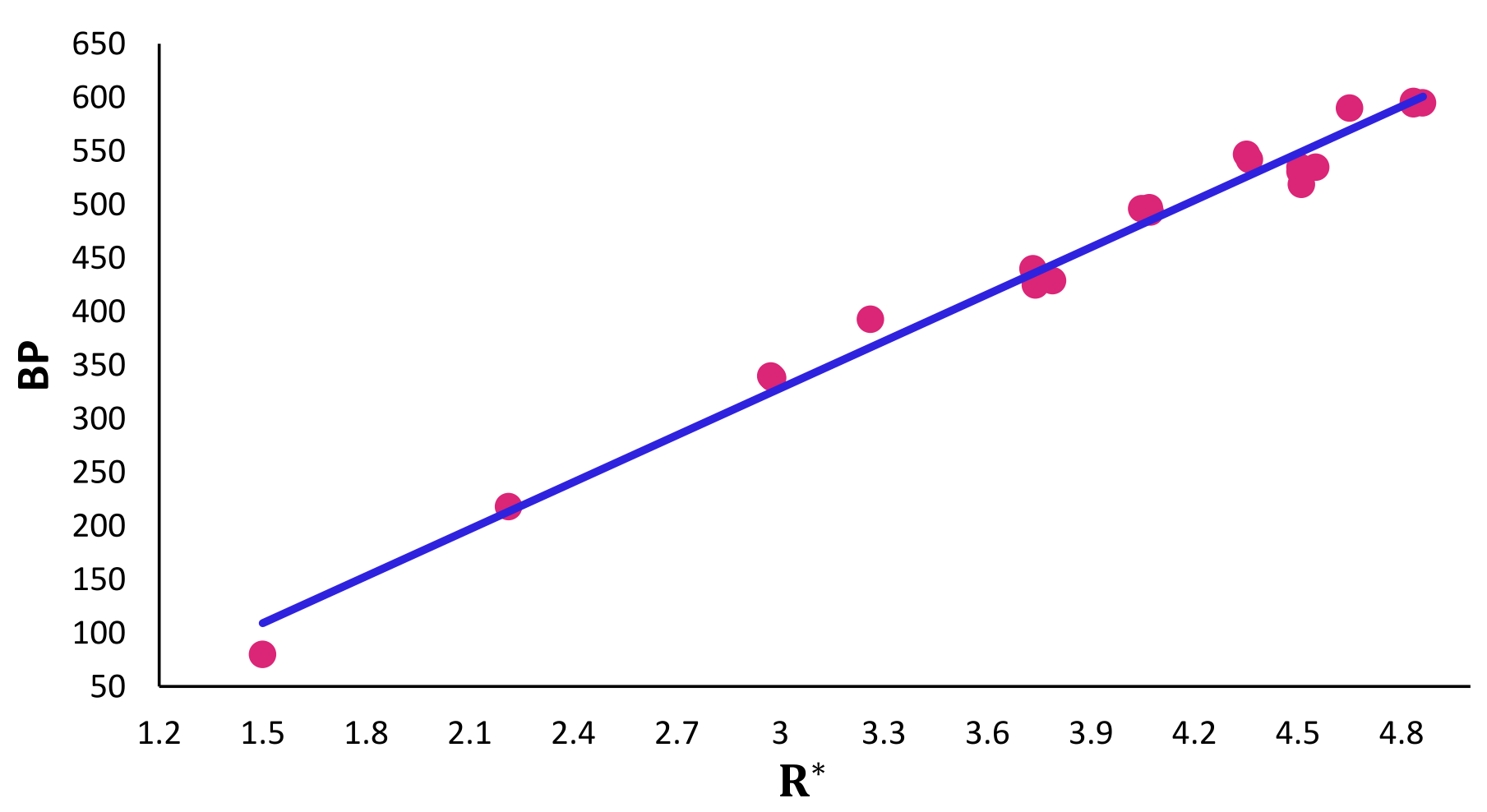}\hspace{1.5cm}
	\includegraphics[width=5.7cm,height=3.4cm]{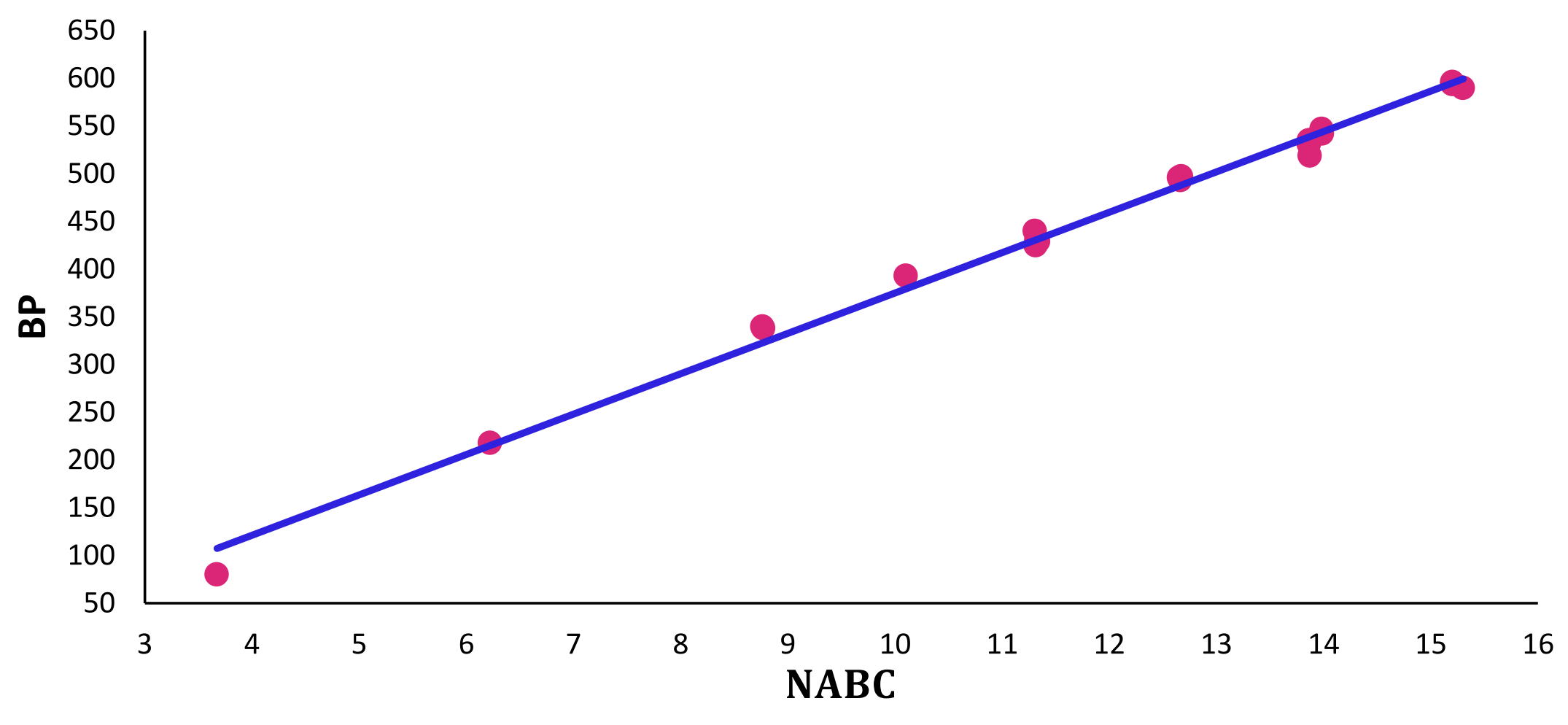}\vspace{0.5cm}
	\includegraphics[width=5.7cm,height=3.4cm]{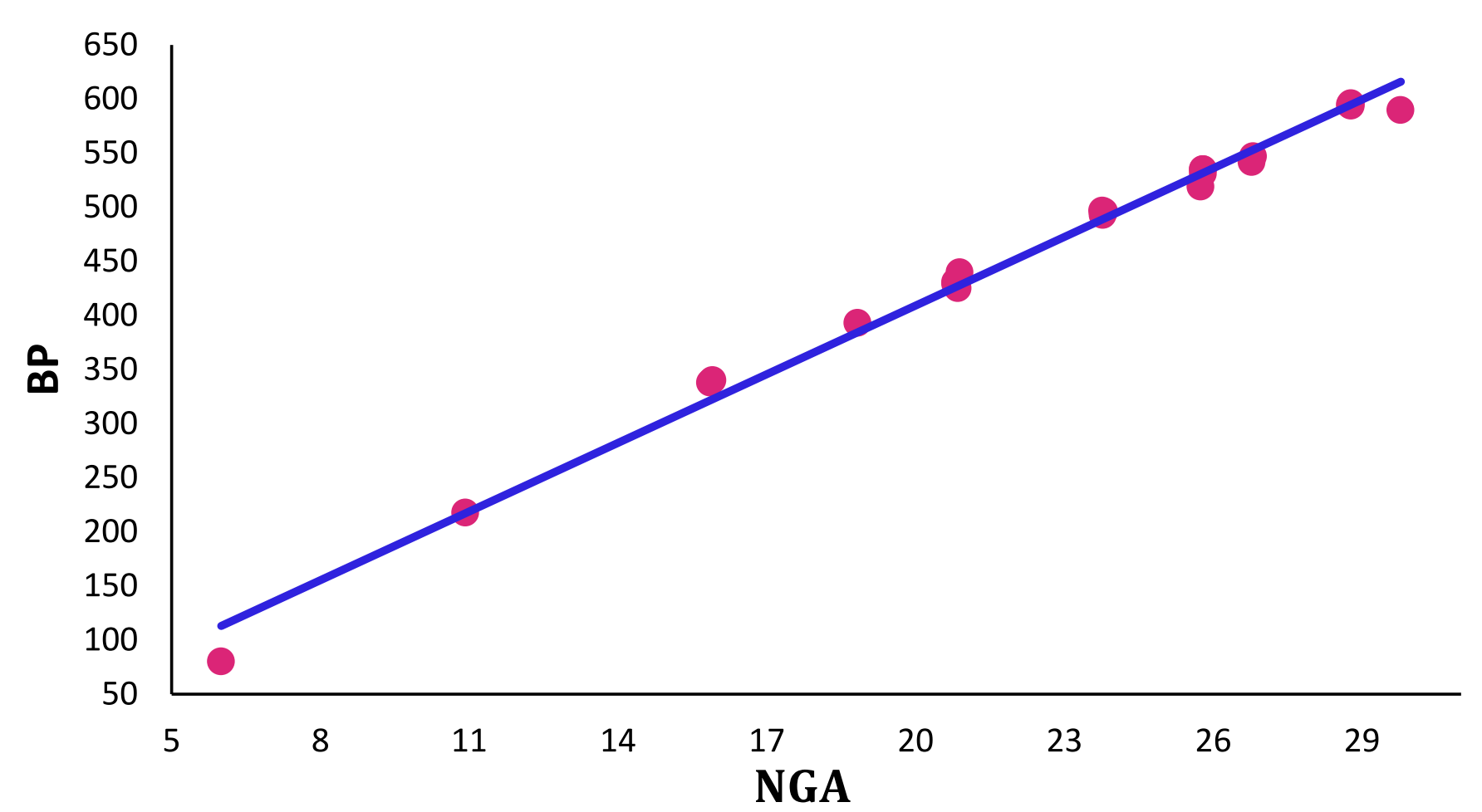}\hspace{1.5cm}
	\includegraphics[width=5.7cm,height=3.4cm]{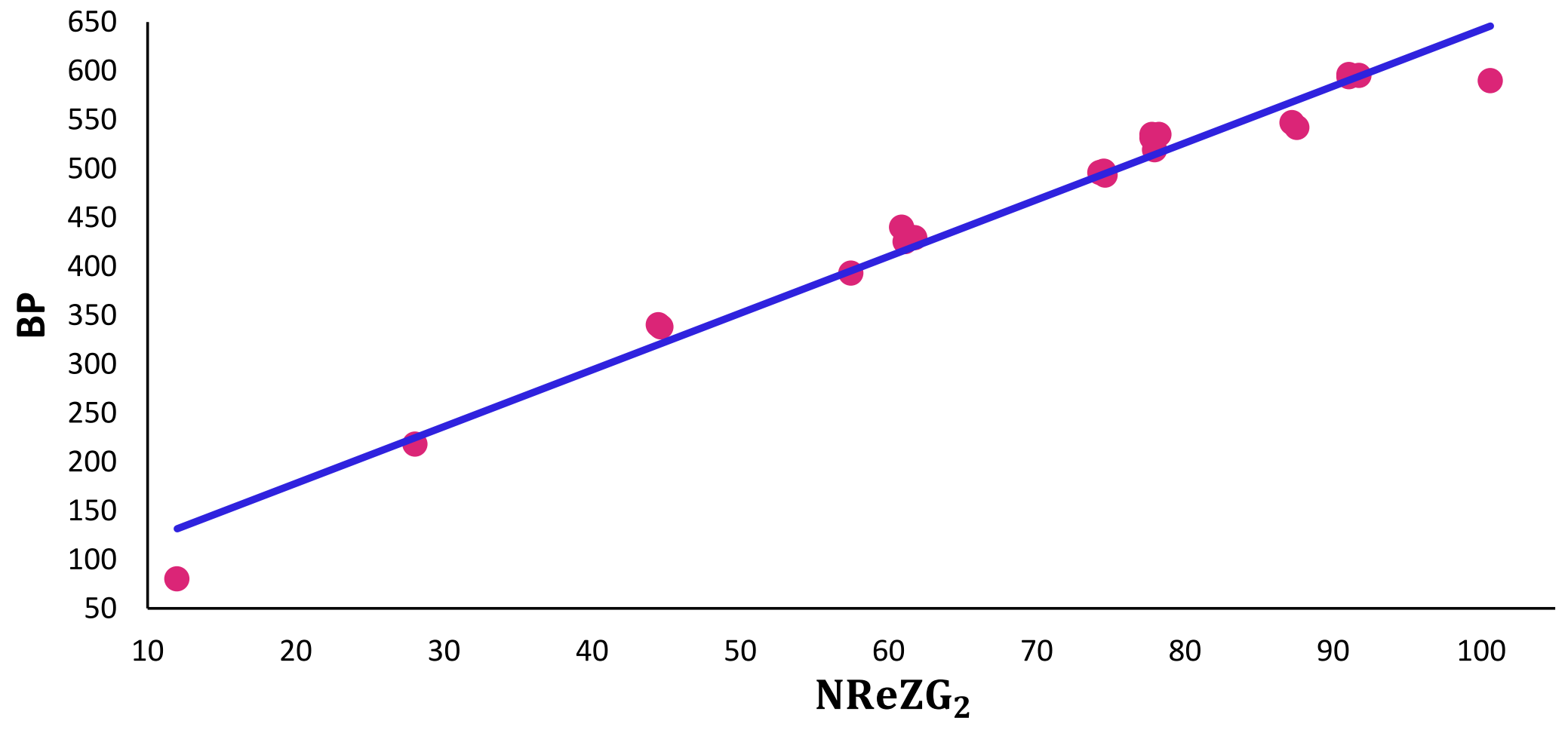}
	\caption{Graphical representation regression model between Boiling point and NDSB version of TIs.}
\end{figure}
\begin{figure}[ht]
	\centering
	\includegraphics[width=5.7cm,height=3.4cm]{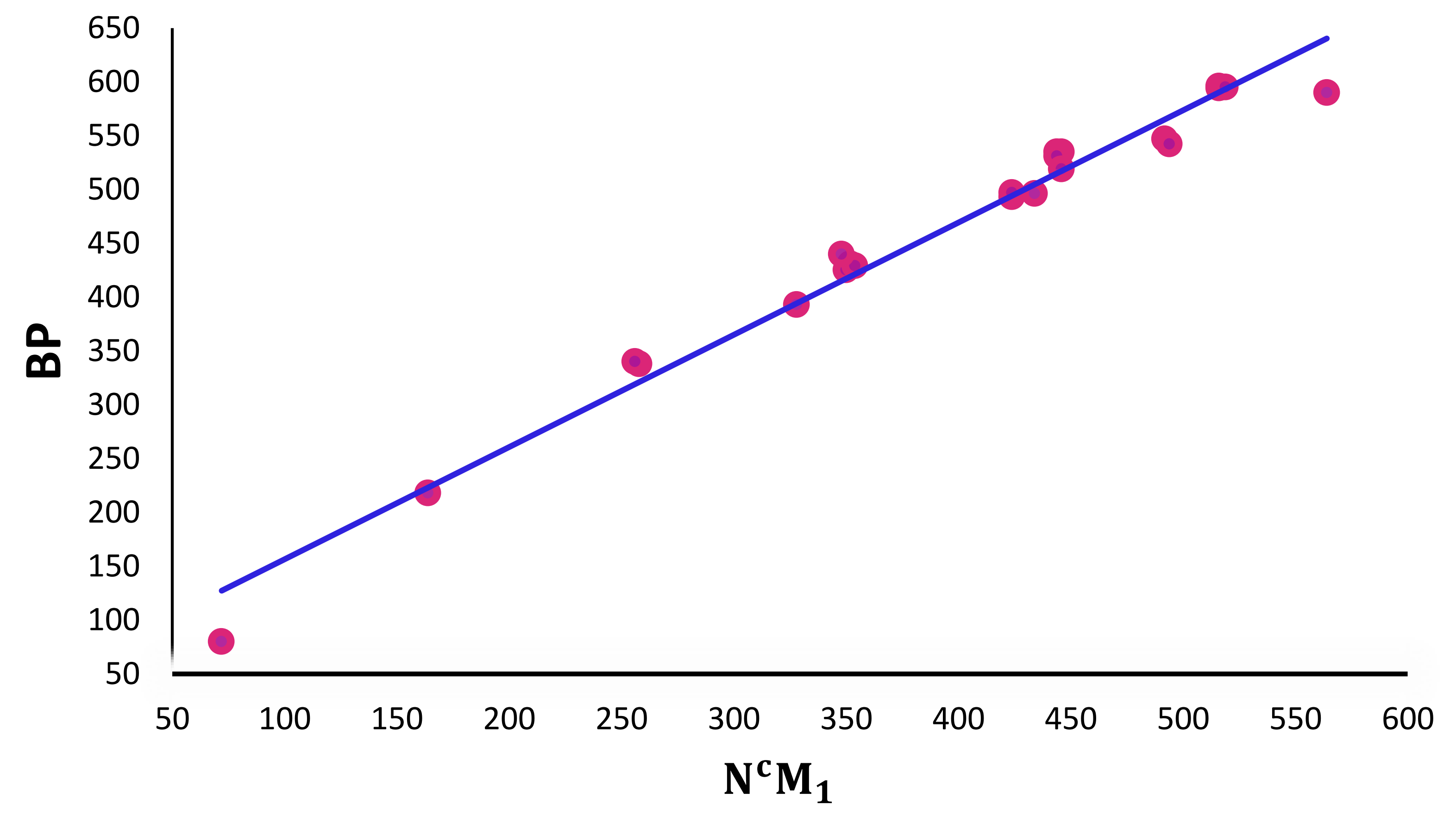}\hspace{1.5cm}
	\includegraphics[width=5.7cm,height=3.4cm]{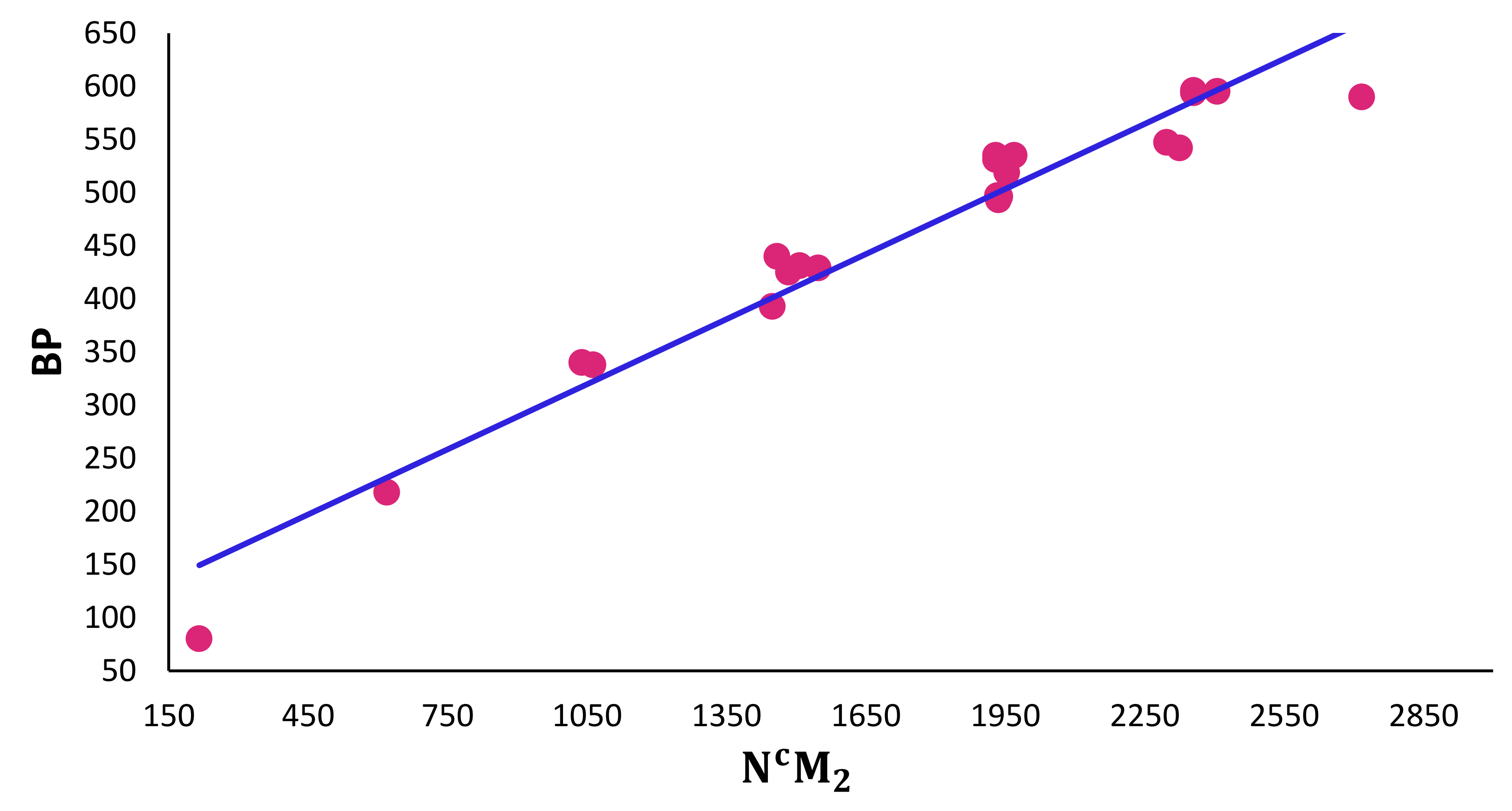}\vspace{0.5cm}
	\includegraphics[width=5.7cm,height=3.4cm]{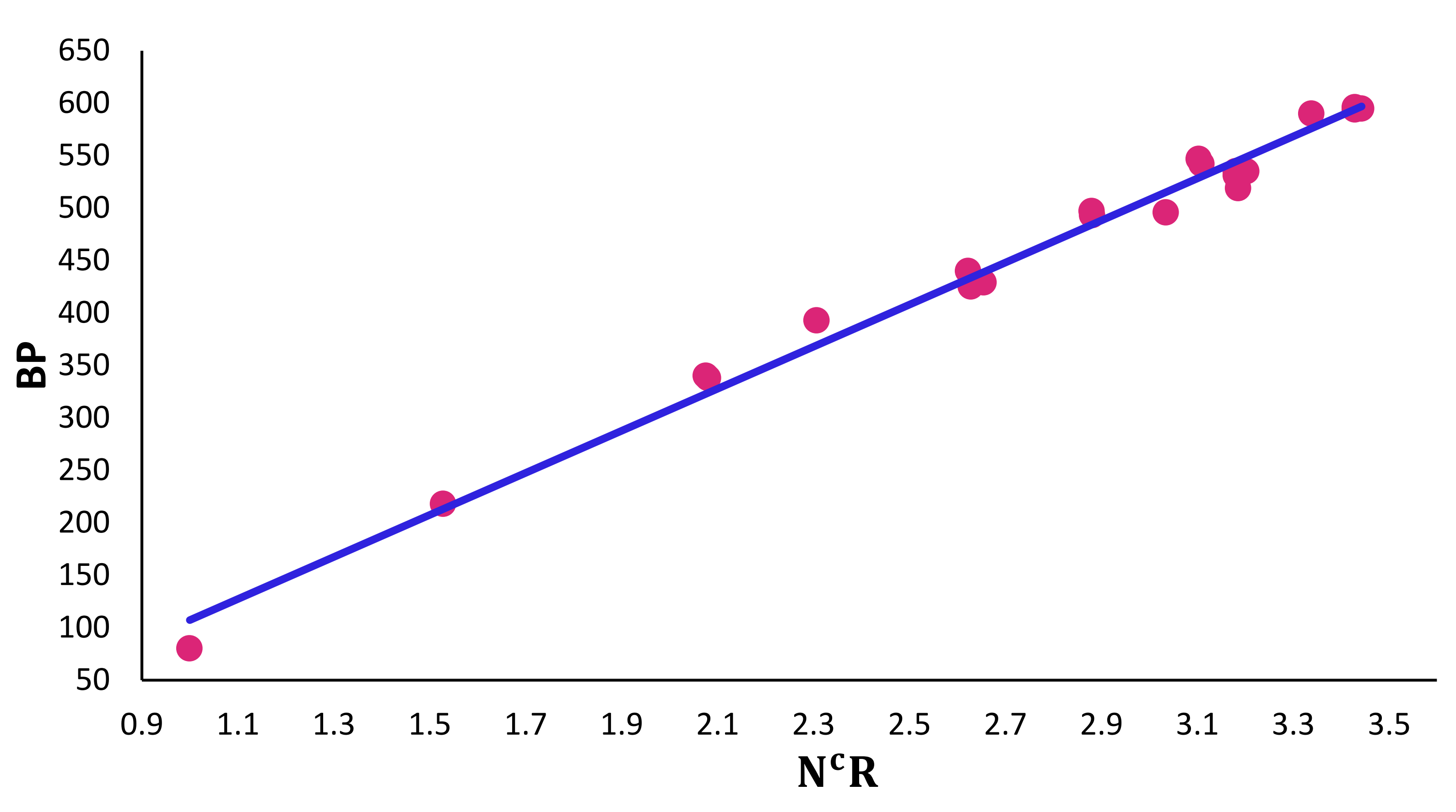}\hspace{1.5cm}
	\includegraphics[width=5.7cm,height=3.4cm]{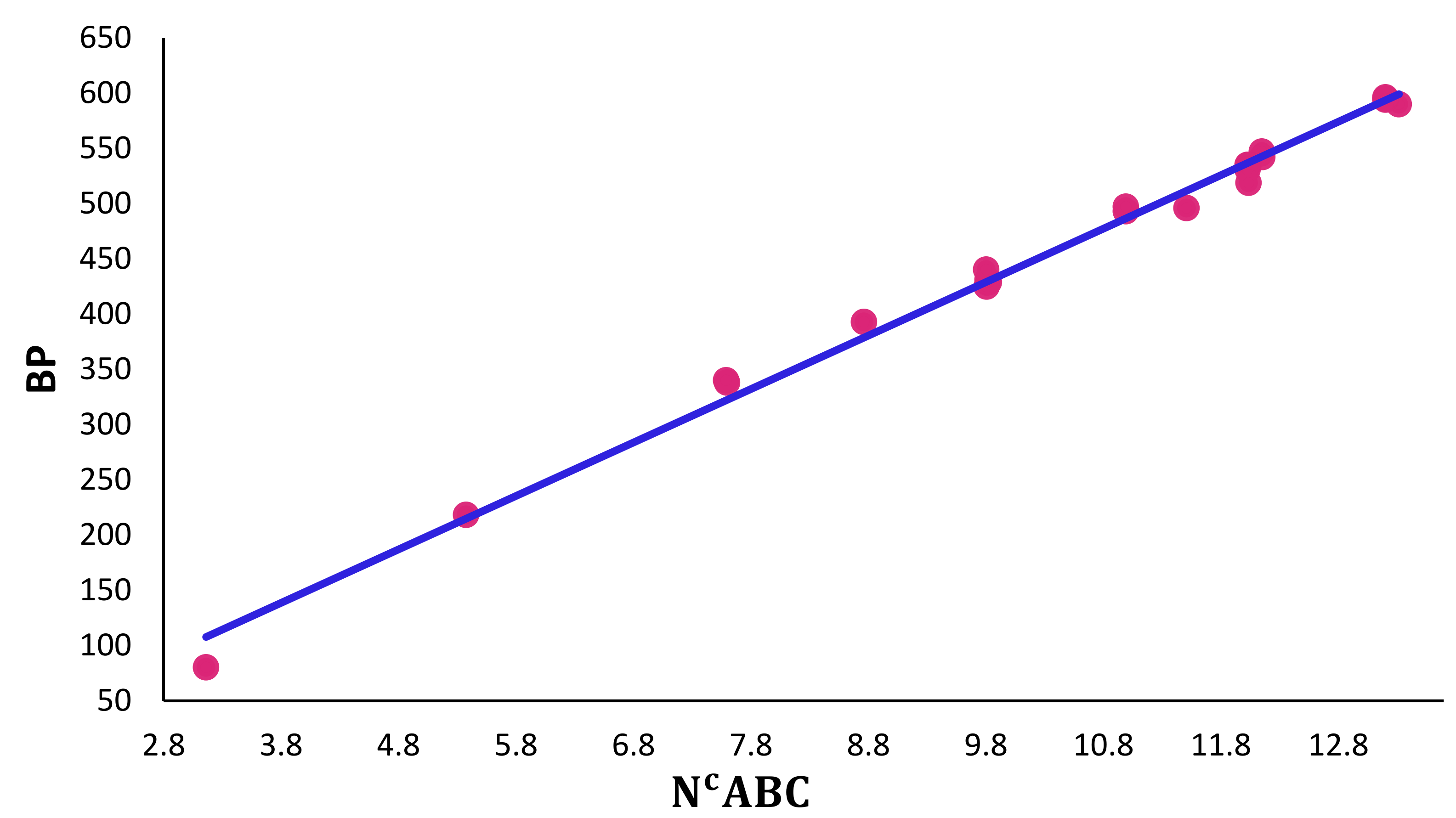}\vspace{0.5cm}
	\includegraphics[width=5.7cm,height=3.4cm]{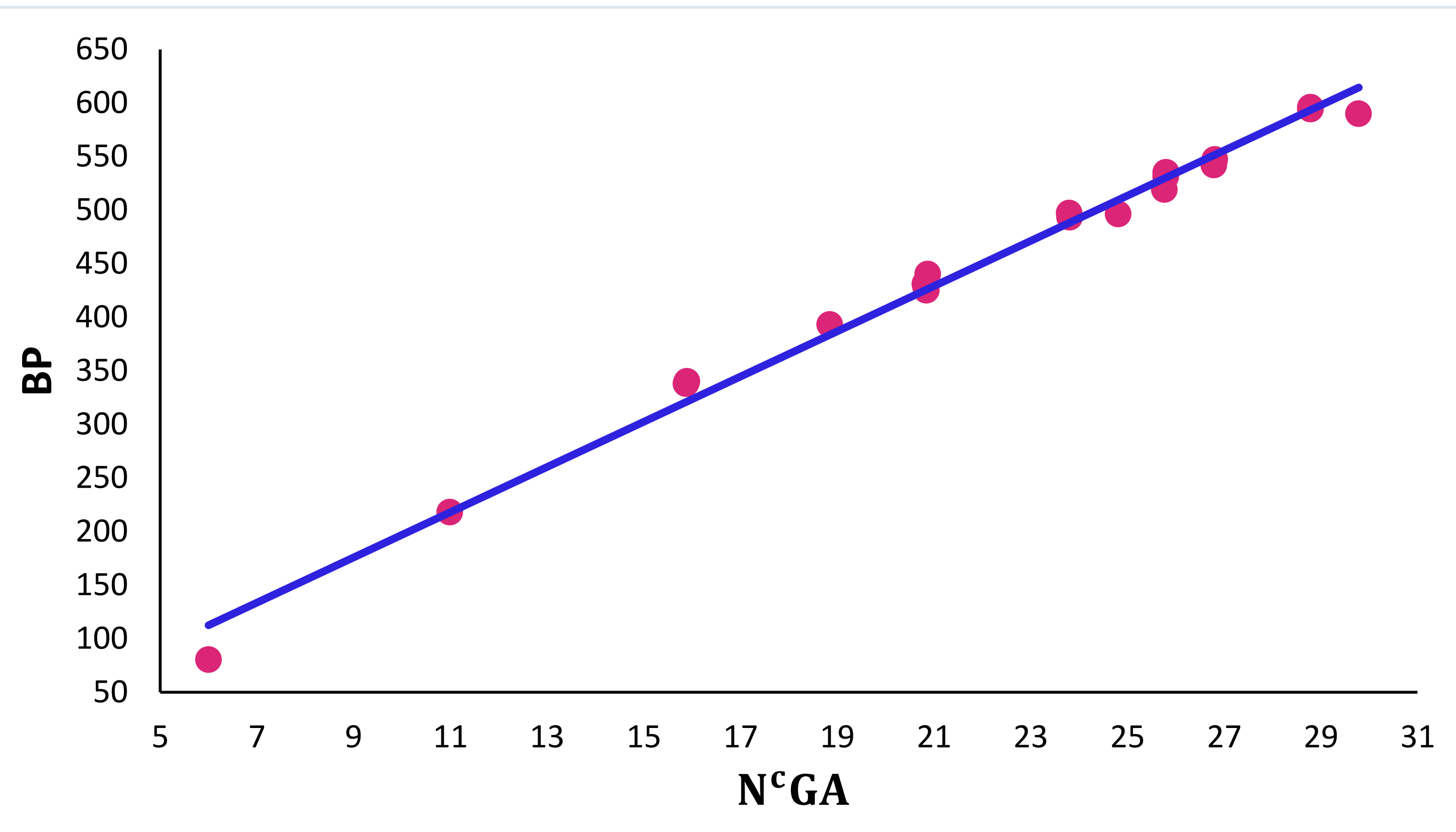}\hspace{1.5cm}
	\includegraphics[width=5.7cm,height=3.4cm]{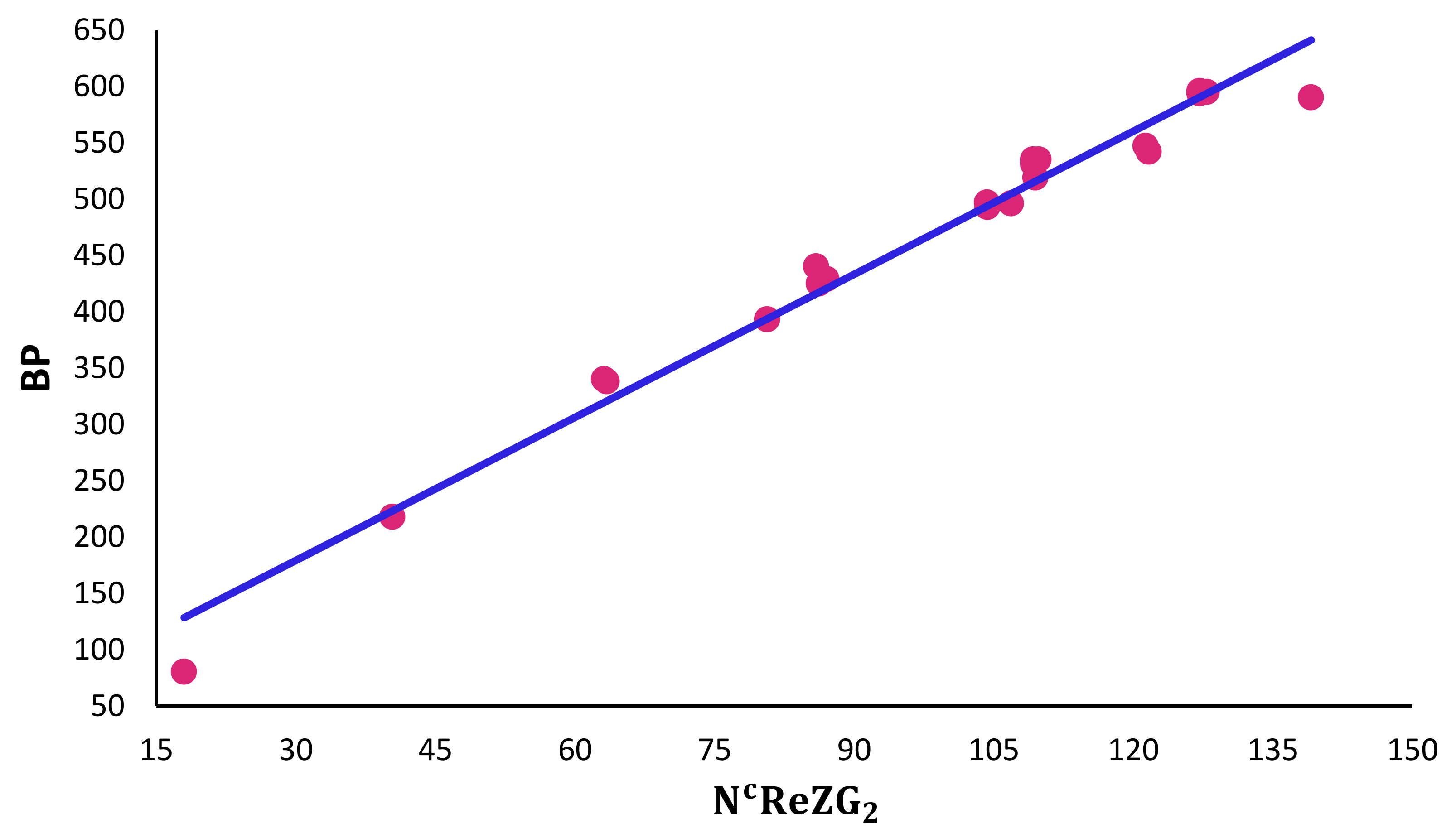}
	\caption{Graphical representation regression model between Boiling point and CNDSB version of TIs.}
\end{figure}

\begin{table}[ht]
	\centering
	\caption{The Significant regression model: Boiling point 22 lower PAHs on Degree-Based based form of TIs.}
	\resizebox{\textwidth}{!}{
		\begin{tabular}{llcc}
			\hline
			Descriptors	& Regression Model & $ R^2 $ & \text{Correlation Coefficient (r)} \\ \hline
			$ M_1 $ & $ BP=3.8778M_1+26.629 $ & 0.9845 & 0.99222 \\
			$ M_2 $ & $ BP=2.8577M_2+61.798 $ & 0.9722 & 0.986002 \\
			$ R $ & $ BP=54.87R-57.555 $ & 0.9926 & 0.996293 \\
			$ ABC $ & $ BP=30.467(ABC)-18.972 $ & 0.9925 & 0.996243 \\
			$ GA $ & $ BP=21.141(GA)-13.635 $ & 0.9911 & 0.99554 \\
			$ ReZG_2 $ & $ BP=15.758(ReZG_2)+26.863 $ & 0.9834 & 0.991665 \\ \hline
	\end{tabular}}
\end{table}

\begin{table}[ht]
	\centering
	\caption{The Significant regression model: Boiling point 22 lower PAHs on Neighborhood Degree Sum-Based form of TIs.}
	\resizebox{\textwidth}{!}{
		\begin{tabular}{llcc}
			\hline
			Descriptors	& \text{Regression Model} & $ R^2 $ & \text{Correlation Coeffecient (r)} \\ \hline
			$ M_1{}^* $ & $ BP=1.4281M_1^*+61.954 $ & 0.9722 & 0.986002 \\
			$ M_2{}^* $ & $ BP=0.3803M_2^*+120.62 $ & 0.9362 & 0.967574 \\
			$ R^* $ & $ BP=146.19R^*-110.17$ & 0.9837  & 0.991817 \\
			$ NABC $ & $ BP=42.312(NABC)-48.038 $ & 0.9933 & 0.996644 \\
			$ NGA $ & $ BP=21.148(NGA)-13.878 $ & 0.9913 & 0.99564 \\
			$ NReZG_2 $ & $ BP=5.8042(NReZG_2)+61.859 $ & 0.9717 & 0.985748 \\ \hline
	\end{tabular}}
\end{table}

\begin{table}[ht]
	\centering
	\caption{The Significant regression model: Boiling point 22 lower PAHs on Closed Neighborhood Degree Sum-Based form of TIs.}
	\resizebox{\textwidth}{!}{
		\begin{tabular}{llcc}
			\hline
			\text{Descriptor}	& \text{Regression Model} & $ R^2 $ & \text{Correlation Coeffecient (r)} \\
			\hline
			$N^cM_1$ & $ BP=1.0427(N^cM_1)+52.261 $ & 0.9758 & 0.987826 \\
			$ N^cM_2 $ & $ BP=0.2041(N^cM_2)+105.26 $ & 0.9473 & 0.973293 \\
			$ N^cR $ &$  BP=200.42(N^cR)-93.144 $ & 0.9872 & 0.993579 \\
			$ N^cABC $ & $ BP=48.386(N^cABC)-45.209 $ & 0.9927 & 0.996343 \\
			$ N^cGA $ & $ BP=21.102(N^cGA)-14.102 $ & 0.9909 & 0.99544 \\
			$ N^cReZG_2 $ & $ BP=4.2328(N^cReZG_2)+52.305 $ & 0.9751 & 0.987472 \\
			\hline
	\end{tabular}}
\end{table}
\begin{figure}[ht]
	\centering
	\includegraphics[width=13cm, height=4.5cm]{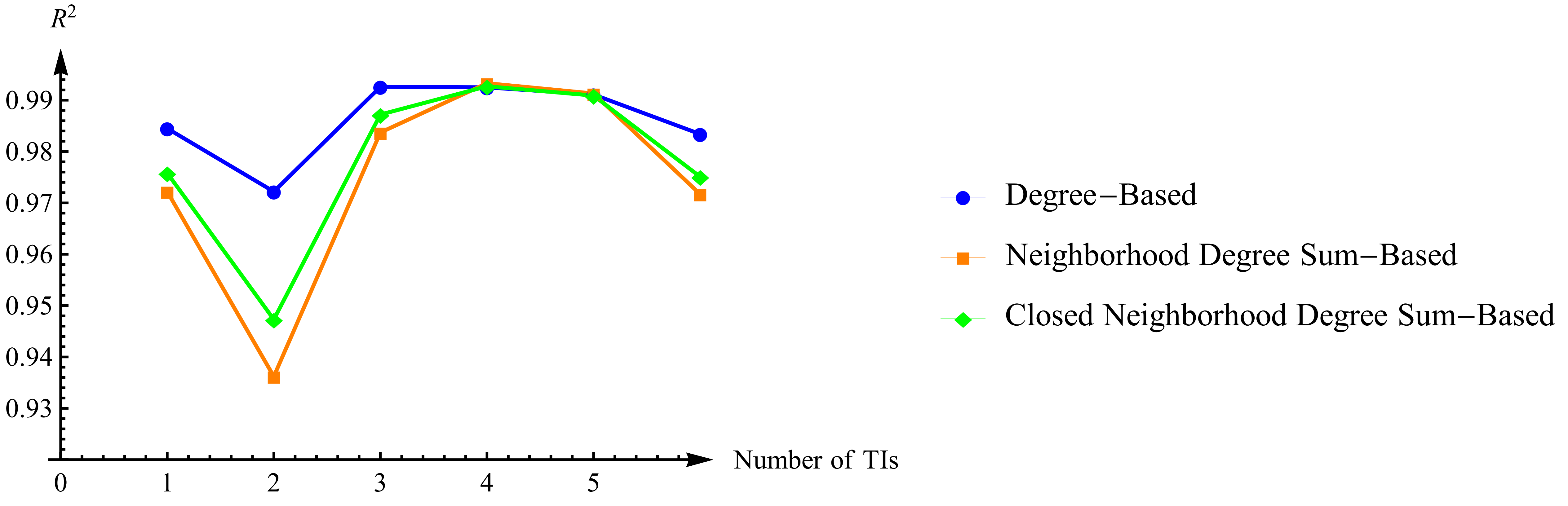}
	\caption{Predictive potential comparison of DB, NDSB and CNDSB versions of considered TIs for BP. }
\end{figure}
\section{Conclusion}
In this manuscript we evaluate the predictive potential of six well known degree dependent TIs for three different versions with normal boiling point of 22 PAHs. We indicate that all considered TIs have excellent correlation with normal boiling point of considered group of molecules. In this article we develop QSPR analysis between three versions of six well known degree dependent TIs normal boiling points of 22 lower PAHs. The linear regression model shows, the Randic, Atom bond connectivity and Geometric Arithmetic Indices ($R, ABC and GA$) have the highest correlation Coefficient values with boiling point in all three versions. The DB version of considered TIs shows better results compared to the NDSB and CNDSB versions, the second best predictor version is CNDSB version. 

\clearpage
\bibliographystyle{plain}
\bibliography{document.bib}

\begin{thebibliography}{10}

\bibitem{abubakar2022neighborhood}
Muhammad~Shafii Abubakar, Kazeem~Olalekan Aremu, and Maggie Aphane.
\newblock Neighborhood versions of geometric--arithmetic and atom bond
  connectivity indices of some popular graphs and their properties.
\newblock {\em Axioms}, 11(9):487, 2022.

\bibitem{chamua2022predictive}
Monjit Chamua, Jibonjyoti Buragohain, A~Bharali, and Mohammad~Essa Nazari.
\newblock Predictive ability of neighborhood degree sum-based topological
  indices of polycyclic aromatic hydrocarbons.
\newblock {\em Journal of Molecular Structure}, 1270:133904, 2022.

\bibitem{dehmer2017quantitative}
Matthias Dehmer, Frank Emmert-Streib, and Yongtang Shi.
\newblock Quantitative graph theory: a new branch of graph theory and network
  science.
\newblock {\em Information Sciences}, 418:575--580, 2017.

\bibitem{deutsch2014m}
Emeric Deutsch and Sandi Klav{\v{z}}ar.
\newblock M-polynomial and degree-based topological indices.
\newblock {\em arXiv preprint arXiv:1407.1592}, 2014.

\bibitem{estrada1998atom}
Ernesto Estrada, Luis Torres, Lissette Rodriguez, and Ivan Gutman.
\newblock An atom-bond connectivity index: modelling the enthalpy of formation
  of alkanes.
\newblock {\em Indian Journal of Chemistry -Section A (IJC-A)}, 37A:849--855,
  1998.

\bibitem{ghorbani2013third}
Modjtaba Ghorbani and Mohammad~A Hosseinzadeh.
\newblock The third version of zagreb index.
\newblock {\em Discrete mathematics, algorithms and applications},
  5(04):1350039, 2013.

\bibitem{gutman2004first}
Ivan Gutman and Kinkar~Ch Das.
\newblock The first zagreb index 30 years after.
\newblock {\em MATCH Commun. Math. Comput. Chem}, 50(1):83--92, 2004.

\bibitem{hayat2020distance}
Sakander Hayat, Suliman Khan, Asad Khan, and Muhammad Imran.
\newblock Distance-based topological descriptors for measuring the
  $\pi$-electronic energy of benzenoid hydrocarbons with applications to carbon
  nanotubes.
\newblock {\em Mathematical Methods in the Applied Sciences}, 2020.

\bibitem{hosamani2017qspr}
Sunilkumar~M Hosamani, Bhagyashri~B Kulkarni, Ratnamma~G Boli, and Vijay~M
  Gadag.
\newblock Qspr analysis of certain graph theocratical matrices and their
  corresponding energy.
\newblock {\em Applied Mathematics and Nonlinear Sciences}, 2(1):131--150,
  2017.

\bibitem{malik2022correlation}
Muhammad Yasir~Hayat Malik, Muhammad~Ahsan Binyamin, and Sakander Hayat.
\newblock Correlation ability of degree-based topological indices for
  physicochemical properties of polycyclic aromatic hydrocarbons with
  applications.
\newblock {\em Polycyclic Aromatic Compounds}, 42(9):6267--6281, 2022.

\bibitem{mondal2019some}
Sourav Mondal, Nilanjan De, and Anita Pal.
\newblock On some new neighborhood degree-based indices for some oxide and
  silicate networks.
\newblock {\em J}, 2(3):384--409, 2019.

\bibitem{mondal2021neighborhood}
Sourav Mondal, Muhammad~Kamran Siddiqui, Nilanjan De, and Anita Pal.
\newblock Neighborhood m-polynomial of crystallographic structures.
\newblock {\em Biointerface Res. Appl. Chem}, 11(2):9372--9381, 2021.

\bibitem{pradeep2017redefined}
R~Pradeep~Kumar, ND~Soner, and MR~Rajesh~Kanna.
\newblock Redefined zagreb, randic, harmonic, ga indices of graphene.
\newblock {\em International Journal of Mathematical Analysis},
  11(10):493--502, 2017.

\bibitem{randic1975characterization}
Milan Randic.
\newblock Characterization of molecular branching.
\newblock {\em Journal of the American Chemical Society}, 97(23):6609--6615,
  1975.

\bibitem{ravi2022closed}
Vignesh Ravi and Kalyani Desikan.
\newblock Closed neighborhood degree sum-based topological descriptors of
  graphene structures.
\newblock {\em Biointerface Res. Appl. Chem}, 12:7111--7124, 2022.

\bibitem{vukivcevic2009topological}
Damir Vuki{\v{c}}evi{\'c} and Boris Furtula.
\newblock Topological index based on the ratios of geometrical and arithmetical
  means of end-vertex degrees of edges.
\newblock {\em Journal of mathematical chemistry}, 46:1369--1376, 2009.

\bibitem{wiener1947structural}
Harry Wiener.
\newblock Structural determination of paraffin boiling points.
\newblock {\em Journal of the American chemical society}, 69(1):17--20, 1947.

\end{thebibliography}

\end{document}